\documentclass[11pt]{article}

\usepackage{a4}
\usepackage{graphicx}
\usepackage[algoruled,linesnumbered,algo2e,vlined]{algorithm2e}
\usepackage[final,mode=multiuser]{fixme}
\usepackage{amsmath,amssymb}
\usepackage{url}
\usepackage{rotating}
\usepackage{comment}
\usepackage{wrapfig}

\newif\ifdraft\drafttrue

\ifdraft
\newcommand{\naricom}[1]{\textcolor{red}{[(成澤)#1]}}

\newcommand{\sadacom}[1]{\textcolor{blue}{[(成定)#1]}}

\newcommand{\todo}[1]{{\color{red}{[ToDo: #1]}}}
\else
\newcommand{\naricom}[1]{}
\newcommand{\inecom}[1]{}
\newcommand{\sadacom}[1]{}
\newcommand{\todo}[1]{}
\fi

\usepackage{array}
\usepackage{amsthm}

\newcommand{\parent}{\mathsf{parent}}

\newcommand{\PLV}{\mathsf{PLV}}
\newcommand{\NLV}{\mathsf{NLV}}
\newcommand{\Pos}{\mathsf{Pos}}
\newcommand{\substr}[3]{#1[#2..#3]}
\newcommand{\rev}[1]{#1^R}

\newcommand{\suftree}[1]{\mathsf{STree}(#1)}
\newcommand{\sa}{\mathsf{SA}}
\newcommand{\rsa}{\mathsf{SA}^{-1}}
\newcommand{\lcp}{\mathsf{LCP}}

\newcommand{\lca}[3]{\mathsf{LCA}_{#1}(#2,#3)}
\newcommand{\nma}[2]{\mathsf{NMA}_{#1}(#2)}
\newcommand{\rmq}[3]{\mathsf{RMQ}_{#1}(#2,#3)}
\newcommand{\armlen}{\mathsf{armlen}}
\newcommand{\lcplen}{\mathsf{lcp}}
\newcommand{\occ}{\mathit{occ}}
\newcommand{\LongestSAGP}{\mathsf{SAGP}}
\newcommand{\LongestSAGPone}{\mathsf{SAGP}_1}
\newcommand{\LongestSAGPtwo}{\mathsf{SAGP}_2}
\newcommand{\Pals}{\mathsf{Pals}}

\newcommand{\op}{\mathsf{op}}
\newcommand{\Active}{\mathsf{Active}}
\newcommand{\Str}{\mathit{str}}

\newcommand{\LeftMost}{\mathsf{LMost}}

\newcommand{\FindR}{\mathsf{findR}}
\newcommand{\FindRarray}{\mathsf{FindR}}
\newcommand{\NextPos}{\mathsf{NextPos}}
\newcommand{\minLin}{\mathit{min}_\mathit{in}}
\newcommand{\minLout}{\mathit{min}_\mathit{out}}

\newcommand{\LPrF}{\mathsf{LPrF}}

\newtheorem{theorem}{Theorem}
\newtheorem{lemma}{Lemma}
\newtheorem{example}{Example}

\usepackage[top=2cm, bottom=2cm, left=2cm, right=2cm, includefoot]{geometry}
\usepackage{authblk}
\usepackage{graphicx}
\usepackage{amsthm}
\usepackage{url}

\usepackage[
type={CC},
modifier={by-nc-nd},
version={4.0},
]{doclicense}

\makeatletter
\def\blfootnote{\xdef\@thefnmark{}\@footnotetext}
\makeatother

\begin{document}
		
		\title{Efficient computation of longest single-arm-gapped palindromes in a string}
		
		\author[1]{Shintaro~Narisada}
		\author[1]{Diptarama~Hendrian}
		\author[1]{Kazuyuki~Narisawa}
		\author[2]{Shunsuke~Inenaga}
		\author[1]{Ayumi~Shinohara}
		\affil[1]{Graduate School of Information Sciences, Tohoku University, Sendai, Japan}
		\affil[2]{Department of Informatics, Kyushu University, Fukuoka, Japan}
		
		\date{}
		
		\maketitle
		
		\begin{abstract}

In this paper, we introduce new types of approximate palindromes called
\emph{single-arm-gapped palindromes} (shortly \emph{SAGPs}).
A SAGP contains a gap in either its left or right arm,
which is in the form of either $wguc \rev{u} \rev{w}$ or
$wuc \rev{u}g\rev{w}$,
where $w$ and $u$ are non-empty strings,
$\rev{w}$ and $\rev{u}$ are respectively the reversed strings of $w$ and $u$,
$g$ is a string called a gap, and $c$ is either a single character or the empty string.
Here we call $wu$ and $\rev{u} \rev{w}$ the arm of the SAGP,
and $|uv|$ the length of the arm.
We classify SAGPs into two groups:
those which have $uc\rev{u}$ as a maximal palindrome (type-1),
and the others (type-2).
We propose several algorithms to compute type-1 SAGPs
with longest arms occurring in a given string, based on suffix arrays.
Then, we propose a linear-time algorithm to compute all type-1 SAGPs
with longest arms, based on suffix trees.
Also, we show how to compute type-2 SAGPs with longest arms
in linear time.
We also perform some preliminary experiments to show practical
performances of the proposed methods.
 
		\end{abstract}
		
		\blfootnote{\textcopyright 2019. This manuscript version is made available under the \doclicenseNameRef license
		\doclicenseImage[imagewidth=2cm]}

	

\section{Introduction}
A palindrome is a string that reads the same forward and backward.
Discovering palindromic structures in strings is a classical and important task in combinatorics on words and string algorithmics
(e.g., see~\cite{DroubayP99,GlenJWZ09,Manacher75,ApostolicoBG95}).
A natural extension to palindromes is to  
allow for a \emph{gap} between the left and right arms of palindromes.
Namely, a string $x$ is called a gapped palindrome 
if $x = wg\rev{w}$ for some strings $w,g$ with $|w| \geq 1$ and $|g| \geq 0$.
For instance, string $\mathtt{abbca | acbba}$ is a palindrome and
string $\mathtt{abbca | baa | acbba}$ is a gapped palindrome with $\mathtt{baa}$ as a gap.
Finding gapped palindromes has applications in bioinformatics,
such as finding secondary structures of RNA sequences called \emph{hairpins}~\cite{Gusfield1997}.
If we further allow for another gap \emph{inside either arm},
then such a palindrome can be written as 
$wg_2ug_1\rev{u}\rev{w}$ or $wug_1\rev{u}g_2\rev{w}$ 
for some strings $w, g_1, g_2, u$ with
$|u| \geq 1$, $|g_1| \geq 0$, $|g_2| \geq 0$, and $|w| \geq 1$.
These types of palindromes characterize
\emph{hairpins with bulges} in RNA sequences,
known to occur frequently in the secondary structures 
of RNA sequences~\cite{ShiWWT14}.
Notice that the special case where $|g_1| \leq 1$ and $|g_2| = 0$
corresponds to usual palindromes,
and the special case where $|g_1| \geq 2$ and $|g_2| = 0$
corresponds to gapped palindromes.

In this paper, we consider a new class of generalized palindromes
where $|g_1| \leq 1$ and $|g_2| \geq 1$,
i.e., palindromes with gaps only \emph{inside} one of its arms.
We call such palindromes as \emph{single-arm-gapped palindromes} (\emph{SAGPs}).
For instance, string $\mathtt{abb | ca | cb | bc | bba}$
is an SAGP of this kind, taking $w = \mathtt{abb}$,
$g_1 = \varepsilon$ (the empty string),
$g_2 = \mathtt{ca}$, and $u = \mathtt{cb}$.

We are interested in occurrences of SAGPs as substrings of a given string $T$.
For simplicity, we will concentrate on
SAGPs with $|g_1| = 0$ containing a gap in their left arms.
However, slight modification of
all the results proposed in this paper can easily be applied to
all the other cases.
For any occurrence of an SAGP $wguu^Rw^R$ beginning at position $b$ in $T$,
the position $b+|wgu|-1$ is called the \emph{pivot} of the occurrence
of this SAGP.
This paper proposes various algorithms to solve
the problem of computing longest SAGPs
for every pivot in a given string $T$ of length $n$.
We classify longest SAGPs into two groups:
those which have $u\rev{u}$ as a maximal palindrome (\emph{type-1}),
and the others (\emph{type-2}).
Firstly, we show a na\"ive $O(n^2)$-time algorithm
for computing type-1 longest SAGPs.
Secondly, we present a simple but practical $O(n^2)$-time algorithm
for computing type-1 longest SAGPs based on
simple scans over the suffix array~\cite{ManberM93}.
We also show that the running time of this algorithm
can be improved by using a dynamic predecessor/successor data structure.
If we employ the van Emde Boas tree~\cite{Boas75},
we achieve $O((n + \occ_1) \log\log n)$-time solution,
where $\occ_1$ is the number of type-1 longest SAGPs to output.
Finally, we present an $O(n+\occ_1)$-time solution
based on the suffix tree data structure~\cite{Weiner73}.
For type-2 longest SAGPs,
we show an $O(n + \occ_2)$-time algorithm,
where $\occ_2$ is the number of type-2 longest SAGPs to output.
Combining the last two results, we obtain an optimal $O(n + \occ)$-time
algorithm for computing all longest SAGPs,
where $\occ$ is the number of outputs.

We performed preliminary experiments to compare practical performances
of our algorithms for finding type-1 longest SAGPs.
We compare the running time of
the na\"ive algorithm, the $O(n^2)$-time suffix array based algorithm,
and the improved suffix array based algorithm with 
several kinds of predecessor/successor data structures in the experiments.

\subsection*{Related work}
For a \emph{fixed} gap length $d$,
one can find all gapped palindromes $wg\rev{w}$ with $|g| = d$
in the input string $T$ of length $n$ in $O(n)$ time~\cite{Gusfield1997}.
Kolpakov and Kucherov~\cite{KolpakovK09} showed 
an $O(n+L)$-time algorithm to compute \emph{long-armed palindromes} in $T$, 
which are gapped palindromes $wg\rev{w}$ such that $|w| \geq |g|$.
Here, $L$ denotes the number of outputs.
They also showed how to compute, in $O(n+L)$ time,
\emph{length-constrained palindromes} which are gapped palindromes
$wg\rev{w}$ such that the gap length $|g|$ is in a predefined range.
Fujishige et al.~\cite{FujishigeNIBT16} proposed online algorithms to compute 
long-armed palindromes and length-constrained palindromes
from a given string $T$.
A gapped palindrome $wg\rev{w}$ is an \emph{$\alpha$-gapped palindrome},
if $|wg| \leq \alpha|w|$ for $\alpha \geq 1$.
Gawrychowski et al.~\cite{GawrychowskiIIK18}
showed that the maximum number of $\alpha$-gapped palindromes
occurring in a string of length $n$ is at most $28\alpha n+7n$.
Very recently, 
I and K\"oppl~\cite{abs-1802-10355} showed an improved bound 
$7(\pi^2/6+1/2) \alpha n - 5n - 1$ for the maximum number of $\alpha$-gapped palindromes.
Since long-armed palindromes are $2$-gapped palindromes for $\alpha = 2$,
$L = O(n)$ and thus Kolpakov and Kucherov's algorithm runs in $O(n)$ time.
Gawrychowski et al.~\cite{GawrychowskiIIK18} also proposed 
an $O(\alpha n)$-time algorithm to compute all $\alpha$-gapped 
palindromes in a given string for any predefined $\alpha \geq 1$.
We emphasize that none of the above algorithms 
can directly be applied to computing SAGPs.

Crochemore et al.~\cite{Crochemore2010} proposed a linear-time algorithm to compute longest previous reverse factor array $\LPrF$ of a string $T$,
where $\LPrF[i]$ stores the longest factor $w$ started at $i$ such that $w^R$
occurs at $T[1..i-1]$ which can be considered as a gapped palindrome.
Later Dumitran et al.~\cite{Dumitran2015, Dumitran2017} studied similar arrays that bound the length of gaps.

A preliminary version of this work appeared in~\cite{Narisada2017}.


\section{Preliminaries}
Let $\Sigma = \{1, \ldots, \sigma\}$ be an integer alphabet of size $\sigma$.
An element of $\Sigma^{*}$ is called a \emph{string}.
For any string $w$, $|w|$ denotes the length of $w$.
The empty string is denoted by $\varepsilon$.
Let $\Sigma^{+} = \Sigma^{*}\setminus\{\varepsilon\}$.
For any $1 \leq i \leq |w|$,
$w[i]$ denotes the $i$-th symbol of $w$.
For a string $w=xyz$, strings $x$, $y$, and $z$ are called 
a \emph{prefix}, \emph{substring}, and \emph{suffix} of $w$, respectively.
The substring of $w$ that begins at position $i$ and ends at position $j$ 
is denoted by $\substr{w}{i}{j}$ for $1 \leq i \leq j \leq |w|$, 
i.e., $\substr{w}{i}{j}=  w[i] \cdots w[j]$.
For $j > i$, let $\substr{w}{i}{j} = \varepsilon$ for convenience.
For two strings $X$ and $Y$,
let $\lcplen(X, Y)$ denote the length of the longest common prefix of $X$ and $Y$.

For any string $x$, let $\rev{x}$ denote 
the reversed string of $x$, i.e. $\rev{x} = x[|x|] \cdots x[1]$.
A string $p$ is called a \emph{palindrome} if $p = \rev{p}$.
Let $T$ be any string of length $n$.
Let $p = T[b..e]$ be a palindromic substring of $T$.
The position $i = \lfloor \frac{b+e}{2} \rfloor$ is called the \emph{center}
of this palindromic substring $p$.
The  palindromic substring $p$ is said to be the \emph{maximal palindrome}
centered at $i$ iff there are no longer palindromes than $p$
centered at $i$, namely,
$T[b-1] \neq T[e+1]$, $b = 1$, or $e = n$.

A string $x$ is called a \emph{single-arm-gapped palindrome}~(SAGP)
if $x$ is in the form of either
$wguc\rev{u}\rev{w}$ or $wuc\rev{u}g\rev{w}$,
with some non-empty strings $w,g,u \in \Sigma^+$
and $c \in \Sigma \cup \{\varepsilon\}$.
For simplicity and ease of explanations,
in what follows we consider only SAGPs whose left arms contain gaps
and $c = \varepsilon$, namely, those of form $wgu \rev{u}\rev{w}$.
But our algorithms to follow can easily be modified to compute
other forms of SAGPs occurring in a string as well.

Let $b$ be the beginning position of an occurrence of 
a SAGP $wgu\rev{u}\rev{w}$ in $T$,
namely $T[b..b+2|wu|+|g|-1] = wgu\rev{u}\rev{w}$.
The position $i = b + |wgu|-1$ is called the \emph{pivot} of this occurrence
of the SAGP.
This position $i$ is also the center of the palindrome $u\rev{u}$.
An SAGP $wgu\rev{u}\rev{w}$ for pivot $i$ in string $T$ is represented by
a quadruple $(i, |w|, |g|, |u|)$ of integers.
In what follows, we will identify the quadruple $(i, |w|, |g|, |u|)$
with the corresponding SAGP $wgu\rev{u}\rev{w}$ for pivot $i$.

For any SAGP $x = wgu\rev{u}\rev{w}$,
let $\armlen(x)$ denote the length of the arm of $x$,
namely, $\armlen(x) = |wu|$.
A substring SAGP $y = wgu\rev{u}\rev{w}$ for pivot $i$ in a string $T$
is said to be a \emph{longest} SAGP for pivot $i$,
if for any SAGP $y'$ for pivot $i$ in $T$, $\armlen(y) \geq \armlen(y')$.

Notice that there can be different choices of
$u$ and $w$ for the longest SAGPs at the same pivot.
For instance, consider string $\mathtt{ccabcabbace}$.
Then, $(7, 1, 3, 2) = \mathtt{c | \underline{abc} | ab | ba | c}$
and $(7, 2, 3, 1) = \mathtt{ca | \underline{bca} | b | b | ac}$ are
both longest SAGPs (with arm length $|wu| = 3$) for the same pivot $7$,
where the underlines represent the gaps.
Of all longest SAGPs for each pivot $i$,
we regard those that have longest palindromes $u\rev{u}$ centered at $i$
as \emph{canonical} longest SAGPs for pivot $i$.
In the above example, $(7, 1, 3, 2) = \mathtt{c | \underline{abc} | ab | ba | c}$
is a canonical longest SAGP for pivot $7$,
while $(7, 2, 3, 1) = \mathtt{ca | \underline{bca} | b | b | ac}$ is not.
Let $\LongestSAGP(T)$ be the set of canonical longest SAGPs for all pivots in $T$.
In this paper, we present several algorithms to compute $\LongestSAGP(T)$.

A position $i$ in string $T$ is said to be
of \emph{type-1} if there exists a SAGP $wgu\rev{u}\rev{w}$
such that $u\rev{u}$ is the maximal palindrome
centered at position $i$,
and is said to be of \emph{type-2} otherwise.
For instance, consider string
$T = \mathtt{baaabaabaacbaabaabac}$ of length $20$.
Position $13$ of $T$ is of type-1,
since there are canonical longest SAGPs 
$(13, 4, 4, 2) = \mathtt{abaa | \underline{baac} | ba | ab | aaba}$ and
$(13, 4, 1, 2) = \mathtt{abaa | \underline{c} | ba | ab | aaba}$
for pivot $13$,
where $\mathtt{ba | ab}$ is the maximal palindrome
centered at position $13$.
On the other hand, Position $6$ of $T$ is of type-2;
the maximal palindrome centered at position $6$
is $\mathtt{aaba | abaa}$ but there are no SAGPs
in the form of $wg \mathtt{aaba | abaa} \rev{w}$ for pivot $6$.
The canonical longest SAGPs for pivot $6$ is
$(6, 1, 1, 3) = \mathtt{a | \underline{a} | aba | aba | a}$.

For an input string $T$ of length $n$ over an integer alphabet of size
$\sigma = n^{O(1)}$,
we perform standard preprocessing which replaces
all characters in $T$ with integers from range $[1, n]$.
Namely, we radix sort the original characters in $T$,
and replace each original character by its rank in the sorted order.
Since the original integer alphabet is of size $n^{O(1)}$,
the radix sort can be implemented with $O(1)$ number of bucket sorts,
taking $O(n)$ total time. Thus, whenever we speak of a string $T$
over an integer alphabet of size $n^{O(1)}$, one can regard $T$ as
a string over an integer alphabet of size $n$.

\vspace*{1ex}
\noindent \textbf{Tools:}
Suppose a string $T$ ends with a unique character that
does not appear elsewhere in $T$.
The \emph{suffix tree}~\cite{Weiner73} of a string $T$,
denoted by $\suftree{T}$, is a
path-compressed trie which represents all suffixes of $T$.
Then, $\suftree{T}$ can be defined as an edge-labeled rooted tree
such that
(1) Every internal node is branching;
(2) The out-going edges of every internal node begin with
    mutually distinct characters;
(3) Each edge is labeled by a non-empty substring of $T$;        
(4) For each suffix $s$ of $T$, there is a unique leaf
    such that the path from the root to the leaf spells out $s$.
It follows from the above definition of $\suftree{T}$ that
if $n = |T|$ then the number of nodes and edges in $\suftree{T}$ is $O(n)$.
By representing every edge label $X$ by a pair $(i, j)$ of integers
such that $X = T[i..j]$, $\suftree{T}$ can be represented with $O(n)$ space.
For a given string $T$ of length $n$
over an integer alphabet of size $\sigma = n^{O(1)}$,
$\suftree{T}$ can be constructed in $O(n)$ time~\cite{Farach-ColtonFM00}.
For each node $v$ in $\suftree{T}$,
let $\Str(v)$ denote the string spelled out from the root to $v$.
According to Property (4),
we sometimes identify each position $i$ in string $T$
with the leaf which represents the corresponding suffix $T[i..n]$.

Suppose the unique character at the end of string $T$
is the lexicographically smallest in $\Sigma$.
The \emph{suffix array}~\cite{ManberM93} of string $T$ of length $n$, 
denoted $\sa_T$, is an array of size $n$ such that 
$\sa_T[i] = j$ iff $T[j..n]$ is the $i$th lexicographically smallest
suffix of $T$ for $1 \leq i \leq n$.
The \emph{reversed suffix array} of $T$,
denoted $\rsa_T$, is an array of size $n$ such that
$\rsa_T[\sa_T[i]] = i$ for $1 \leq i \leq n$.
The \emph{longest common prefix array} of $T$, denoted $\lcp_T$,
is an array of size $n$ such that
$\lcp_T[1] = -1$ and $\lcp_T[i] = \lcplen(T[\sa_T[i-1]..n], T[\sa_T[i]..n])$
for $2 \leq i \leq n$.
The arrays $\sa_T$, $\rsa_T$, and $\lcp_T$ for 
a given string $T$ of length $n$ over an integer alphabet of size $\sigma = n^{O(1)}$
can be constructed in $O(n)$ time~\cite{KarkkainenSB06,KasaiLAAP01}.

For a rooted tree $\mathcal{T}$, the lowest common ancestor $\lca{\mathcal{T}}{u}{v}$ of two nodes $u$ and $v$ 
in $\mathcal{T}$ is the deepest node in $\mathcal{T}$ which has $u$ and $v$ as its descendants.
It is known that after a linear-time preprocessing on the input tree,
querying $\lca{\mathcal{T}}{u}{v}$ for any two nodes $u,v$ can be answered 
in constant time~\cite{bender2000lca}.

Consider a rooted tree $\mathcal{T}$
where each node is either marked or unmarked.
For any node $v$ in $\mathcal{T}$,
let $\nma{\mathcal{T}}{v}$ denote the deepest marked ancestor of $v$.
There exists a linear-space algorithm which marks any unmarked node and 
returns $\nma{\mathcal{T}}{v}$ for any node $v$ in amortized $O(1)$ time~\cite{westbrook1992fast}.

Let $A$ be an integer array of size $n$.
A range minimum query $\rmq{A}{i}{j}$ of a given pair $(i,j)$ of
indices~($1 \leq i \leq j \leq n$) asks an
index $k$ in range $[i, j]$ which stores the minimum value in $A[i..j]$.
After $O(n)$-time preprocessing on $A$,
$\rmq{A}{i}{j}$ can be answered in $O(1)$ time
for any given pair $(i, j)$ of indices~\cite{bender2000lca}.

Let $S$ be a set of $m$ integers from universe $[1,n]$,
where $n$ fits in a single machine word.
A predecessor (resp. successor) query for
a given integer $x$ to $S$ 
asks the largest (resp. smallest) value in $S$ 
that is smaller (resp. larger) than $x$.
Let $u(m,n)$, $q(m,n)$ and $s(m,n)$ denote 
the time for updates (insertion/deletion) of elements,
the time for predecessor/successor queries,
and the space of a dynamic predecessor/successor data structure.
Using a standard balanced binary search tree,
we have $u(m,n) = q(m,n) = O(\log m)$ time
and $s(n,m) = O(m)$ space.
The Y-fast trie~\cite{Willard83} achieves 
$u(m,n) = q(m,n) = O(\log \log m)$ \emph{expected} time
and $s(n,m) = O(m)$ space,
while the van Emde Boas tree~\cite{Boas75} does 
$u(m,n) = q(m,n) = O(\log \log m)$ \emph{worst-case} time
and $s(n,m) = O(n)$ space.

\section{Algorithms for computing canonical longest SAGPs}

In this section, we present
several algorithms to compute
the set $\LongestSAGP(T)$ of canonical longest SAGPs
for all pivots in a given string $T$.

Let $\Pos_1(T)$ and $\Pos_2(T)$ be the sets of type-1 and type-2 positions in $T$,
respectively.
Let $\LongestSAGP(T, i)$
be the subset of $\LongestSAGP(T)$
whose elements are canonical longest SAGPs for pivot $i$.
Let $\LongestSAGPone(T) = \bigcup_{i \in \Pos_1(T)} \LongestSAGP(T, i)$
and $\LongestSAGPtwo(T) = \bigcup_{i \in \Pos_2(T)} \LongestSAGP(T, i)$.
Clearly $\LongestSAGPone(T) \cup \LongestSAGPtwo(T) = \LongestSAGP(T)$
and $\LongestSAGPone(T) \cap \LongestSAGPtwo(T) = \emptyset$.
The following lemma gives an useful property to
characterize the type-1 positions of $T$.
\begin{lemma} \label{lem:longest_maximal_SAGP}
	Let $i$ be any type-1 position of a string $T$ of length $n$.
	Then, a SAGP $wgu\rev{u}\rev{w}$ is a canonical longest SAGP for pivot $i$
	iff $u\rev{u}$ is the maximal palindrome
	centered at $i$ and $\rev{w}$ is the longest non-empty prefix of
	$T[i+|\rev{u}|+1..n]$ such that $w$ occurs at least once in $T[1..i-|u|-1]$.
\end{lemma}
\begin{proof}
	($\Rightarrow$)
	Assume on the contrary that $u\rev{u}$ is not the maximal
	palindrome centered at $i$,
	and let $xu\rev{u}\rev{x}$ be the maximal palindrome centered
	at position $i$ with $|x| \geq 1$.
	If $\rev{w} = \rev{x}$,
	then since position $i$ is of type-1,
	there must be a SAGP $w'g'xu\rev{u}\rev{x}\rev{w'}$
	with $|w'| \geq 1$ for pivot $i$,
	but this contradicts that $wgu\rev{u}\rev{w}$ is a longest SAGP for pivot $i$.
	Hence $\rev{x}$ is a proper prefix of $\rev{w}$.
	See Figure~\ref{fig:longest_SAGP_necessary}.
	Let $\rev{x}\rev{w''} = \rev{w}$.
	Since $\rev{w''}$ is a non-empty suffix of $\rev{w}$,
	$w''$ is a non-empty prefix of $w$.
	This implies that there exists a SAGP
	$w''g''xu \rev{u}\rev{x} \rev{w''}$ for pivot $i$.
	However, this contradicts that $wgu\rev{u}\rev{w}$ is
	a canonical longest SAGP for pivot $i$.
	Consequently, $u\rev{u}$ is the maximal palindrome centered at $i$,
	and now it immediately follows that $\rev{w}$ is the 
	longest non-empty prefix of $T[i+|\rev{u}|+1..n]$ such that
	$w$ occurs at least once in $T[1..i-|u|-1]$.
	
	\begin{figure}[t]
		\centerline{
			\includegraphics[scale=0.33, clip]{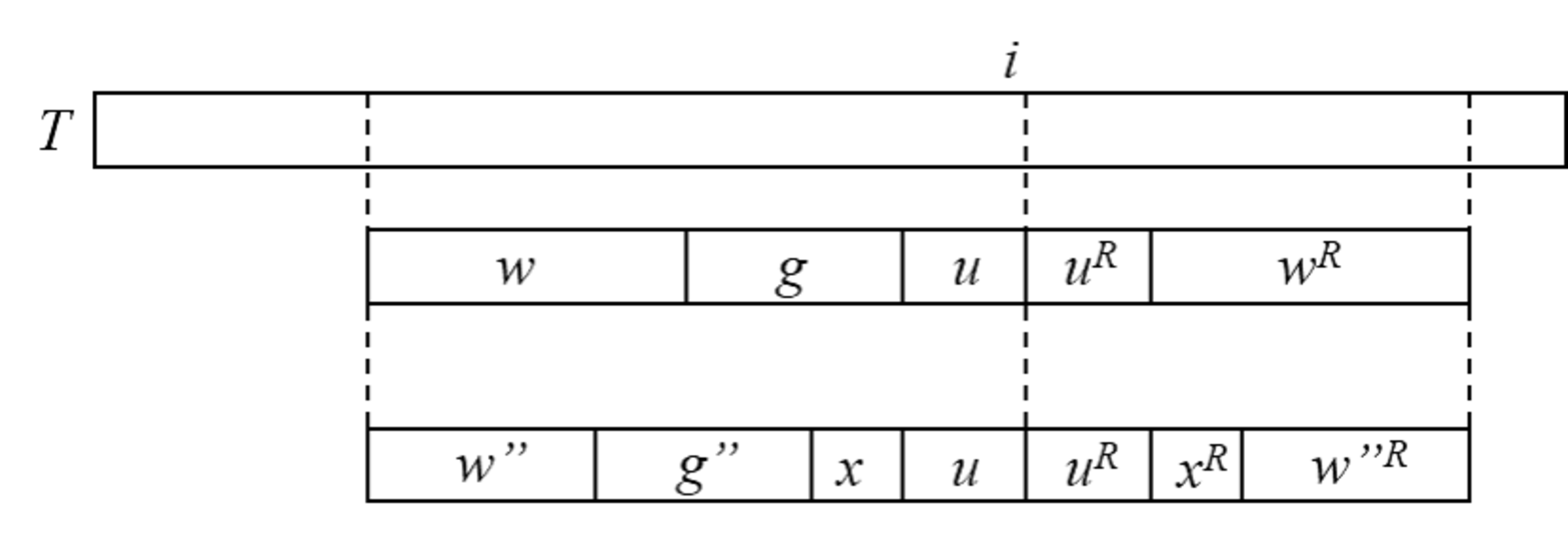}
		}
		\caption{
			Illustration for a necessary condition for a canonical longest SAGP
			(proof of ($\Rightarrow$) for Lemma~\ref{lem:longest_maximal_SAGP}):
			$wgu\rev{u}\rev{w}$ is a canonical longest SAGP
			for pivot $i$.
			For the same pivot $i$,
			there cannot exist a SAGP $w''g''xu\rev{u}\rev{x}\rev{w''}$
			where $xu\rev{u}\rev{x}$ is the maximal palindrome centered at $i$
			and $w''$ is a prefix of $w$,
			since it contradicts that
			$wgu\rev{u}\rev{w}$ is a canonical longest SAGP for $i$.}
		\label{fig:longest_SAGP_necessary}
	\end{figure}
	
	($\Leftarrow$)
	First, we show that $wgu\rev{u}\rev{w}$ is a longest
	SAGP for pivot $i$.
	See Figure~\ref{fig:longest_SAGP_sufficient}.
	Let $u'$ be any proper suffix of $u$,
	and assume on the contrary that there exists
	a SAGP $w'g'u'\rev{u'}\rev{w'}$ for pivot $i$
	such that $|w'u'| > |wu|$.
	Since $|u'| < |u|$, the occurrence of $\rev{w}$ at position $i+|\rev{u}|$
	is completely contained in the occurrence of $\rev{w'}$
	at position $i+|\rev{u'}|$.
	This implies that any occurrence of $w'$ to the left of $u'\rev{u'}$
	completely contains an occurrence of $w$,
	reflected from the occurrence of $\rev{w}$ in $\rev{w'}$.
	However, 
	the character $a$ that immediately precedes the occurrence of $w$ in $w'$
	must be distinct from the character $b$ that immediately follows $\rev{w}$,
	namely $a \neq b$.
	This contradicts that $w'g'u'\rev{u'}\rev{w'}$ is a SAGP for pivot $i$.
	Hence, $wgu\rev{u}\rev{w}$ is a longest SAGP for pivot $i$.
	Since $u\rev{u}$ is the maximal palindrome centered at $i$,
	we cannot extend $u$ to its left nor $\rev{u}$ to its right
	for the same center $i$.
	Thus, $wgu\rev{u}\rev{w}$ is a canonical longest SAGP for pivot $i$.
\end{proof}

\begin{figure}[t]
	\centerline{
		\includegraphics[scale=0.33, clip]{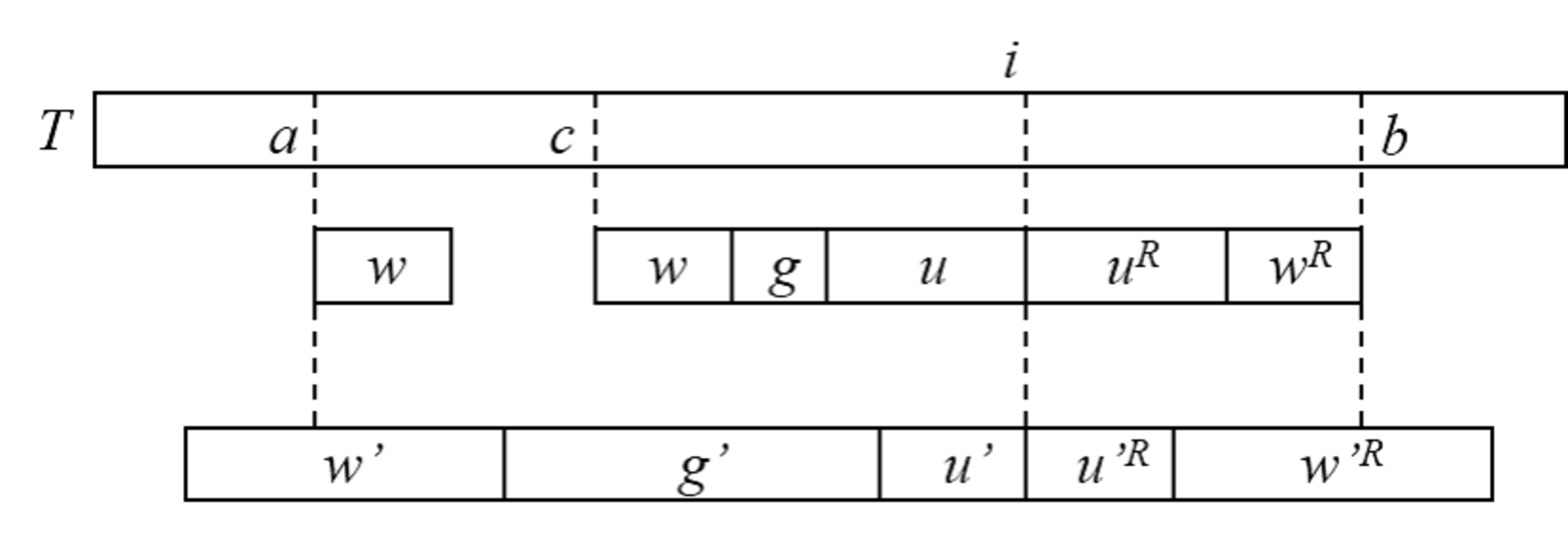}
	}
	\caption{
		Illustration for a sufficient condition for a canonical longest SAGP
		(proof of ($\Leftarrow$) for Lemma~\ref{lem:longest_maximal_SAGP}):
		$u\rev{u}$ is the maximal palindrome centered at $i$
		and $\rev{w}$ is the longest prefix of $T[i+|\rev{u}\rev{w}|+1..n]$
		such that $w$ occurs at least once in $T[1..i-|u|-1]$,
		and thus $c \neq b$.
		Then, there cannot exist a longer SAGP $w'g'u'\rev{u'}\rev{w'}$
		for the same pivot $i$, since $a \neq b$.}
	\label{fig:longest_SAGP_sufficient}
\end{figure}

 We define two arrays $\Pals$ and $\LeftMost$ as follows:
\begin{align*}
 \Pals[i] &= \{r \mid T[i-r+1..i+r] \mbox{ is a maximal palindrome in }   
 T \mbox{ for pivot } i \}. \\
  \LeftMost[c] &= \min\{i \mid T[i]=c \} \mbox{ for } c \in \Sigma.
\end{align*}
By Lemma~\ref{lem:longest_maximal_SAGP}, a position $i$ in $T$ is
of type-1 iff $\LeftMost[T[i+\Pals[i]+1]] < i-\Pals[i]$.

\begin{lemma} \label{lem:determine_types}
	Given a string $T$ of length $n$ over an integer alphabet
	of size $n^{O(1)}$,
	we can determine whether each position $i$ of $T$ is of type-1
	or type-2 in a total of $O(n)$ time and space.
\end{lemma}

\begin{proof}
	Let $u\rev{u}$ be the maximal palindrome centered at $i$.
	Observe that, by Lemma~\ref{lem:longest_maximal_SAGP},
	$i$ is a type-1 position iff
	the character $a = T[i+|\rev{u}|+1]$ which immediately
	follows $\rev{u}$ occurs in $T[1..i-|u|-1]$.
	Let $\Sigma_T$ be the set of distinct characters occurring in $T$.
	We construct an array $\LeftMost$ of size $|\Sigma_T|$
	such that for each $1 \leq j \leq |\Sigma_T|$,
	$\LeftMost[j]$ stores the leftmost occurrence of the lexicographically
	$j$th character in $T$.
	Using the above observation and the array $\LeftMost$,
	we can determine in $O(1)$ time whether 
	a given position $i$ of $T$ is of type-1 or type-2
	by $\LeftMost[T[i+\Pals[i]+1]] < i-\Pals[i]$.
	We can sort the characters in $\Sigma_T$ in $O(n)$
	by constructing $\sa_T$ in $O(n)$ time and space.
\end{proof}

By Lemma~\ref{lem:longest_maximal_SAGP} and Lemma~\ref{lem:determine_types},
we can consider an algorithm to compute $\LongestSAGP(T)$
by computing $\LongestSAGPone(T)$ and $\LongestSAGPtwo(T)$ separately,
as shown in Algorithm~\ref{alg:compute_SAGP}.
In this algorithm, we also construct an auxiliary array $\NextPos$ defined by
	\( \NextPos[i] = \min \{ j \mid i < j, \ T[i] = T[j] \}\) for each $1 \leq i \leq n$,
	which will be used in Section~\ref{sec:type-2}.
	
\begin{lemma} \label{lem:correctness_alg:compute_SAGP}
	Algorithm~\ref{alg:compute_SAGP} correctly computes $\LongestSAGP(T)$.
\end{lemma}

\begin{proof}	
	In line~\ref{alg:compute_SAGP:line:pals}, we firstly compute an array $\Pals$.
	$\Pals[i]$ stores radius $r$ of maximal palindrome centered at $i$.
	We can compute $\Pals$ in $O(n)$ time and space 
	applying Manacher's algorithm~\cite{Manacher75}.
	We show how to compute $\Pals$ in Algorithm~\ref{alg:compute_Pals}.
	In the first {\bf for}-loop, we construct auxiliary arrays
	$\LeftMost$ and $\NextPos$.
	The correctness of the computation of these arrays is obvious.
	We use $\NextPos$ when computing $\LongestSAGPtwo$.
	In line~\ref{alg:compute_SAGP:line:if},
	since we correctly determine which each position of $T$ is of type-1 or type-2
	by Lemma~\ref{lem:determine_types},
	we must compute $\Pos_1(T)$ and $\Pos_2(T)$ in the second {\bf for}-loop.
	Therefore, by referring each element of $\Pos_1(T)$ and $\Pos_2(T)$ respectively,
	we can compute $\LongestSAGPone(T)$ and $\LongestSAGPtwo(T)$,
	namely $\LongestSAGP(T)$.
\end{proof}	

 \begin{algorithm2e}[t]
 	\caption{computing $\LongestSAGP(T)$}
 	\label{alg:compute_SAGP}
 	\KwIn{string $T$ of length $n$}\KwOut{$\LongestSAGP(T)$}
 	compute $\Pals$;  \qquad\qquad\qquad\quad \textit{/* Algorithm~\ref{alg:compute_Pals} */} \label{alg:compute_SAGP:line:pals}\\
 	\For{$i=n$ {\bf downto} $1$}{
 		$c = T[i]$;
 		$\NextPos[i] = \LeftMost[c]$;
 		$\LeftMost[c] = i$\;
 	}
 	\For{$i=1$ {\bf to} $n$}{
 		\uIf{$\LeftMost[T[i+\Pals[i]+1]] < i-\Pals[i]$}
 		{ $\Pos_1(T) = \Pos_1(T) \cup \{ i\}$; \qquad  \textit{/* position $i$ is of type-1  */} }  \label{alg:compute_SAGP:line:if}
 		\Else{$\Pos_2(T) = \Pos_2(T) \cup \{ i\}$; \qquad  \textit{/* position $i$ is of type-2  */}   }
 	}
 	 	compute $\LongestSAGPone(T)$; \qquad\qquad \textit{/* Section~\ref{sec:type-1} */} 	\label{alg:compute_SAGP:line:computeOne} \\
 	 	compute $\LongestSAGPtwo(T)$; \qquad\qquad \textit{/* Section~\ref{sec:type-2} */ }\label{alg:compute_SAGP:line:computeTwo}
 \end{algorithm2e}

\begin{algorithm2e}[t]
	\caption{computing $\Pals$  \qquad \textit{/* proposed by Manacher~\cite{Manacher75} */}}
	\label{alg:compute_Pals}
	\KwIn{string $T$ of length $n$}
	\KwOut{$\Pals$ \textit{/*$\Pals[i]$: maximal even palindrome for pivot $i$*/}}
	$\Pals[0] = 0$\;
	$i = 2$; \ $c = 1$; \ $r = 0$\;
	\While{$c \leq n$}{
		$j = 2*c-j+1$\;
		\While{$T[i]=T[j]$}{
			$i=i+1$; \ $j=j-1$; \ $r=r+1$\;
		}
		$\Pals[c] = r$\;
		$d = 1$\;
		\While{$d \leq r$}{
			$r_l = \Pals[c-d]$\;
			\lIf{$r_l =r-d$}{\textbf{break}} 
			$\Pals[c+d] = \min\{r-d, \ r_l\}$\;
			$d = d+1$\;
		}
		\lIf{$d>r$}{$i = i+1$; \ $r = 0$} 
		\lElse{$r = r_l$}
		$c = c + d$\;
	}
\end{algorithm2e}

In the following subsections, we present algorithms to compute
$\LongestSAGPone(T)$ and $\LongestSAGPtwo(T)$ respectively,
assuming that the arrays $\Pals$, $\LeftMost$ and $\NextPos$ have already been computed.

\subsection{Computing $\LongestSAGPone(T)$ for type-1 positions} \label{sec:type-1}

In what follows, we present several algorithms 
corresponding to the line~\ref{alg:compute_SAGP:line:computeOne} in Algorithm~\ref{alg:compute_SAGP}.
Lemma~\ref{lem:longest_maximal_SAGP} allows us greedy strategies 
to compute the longest prefix $\rev{w}$ of $T[i+\Pals[i]+1..n]$
such that $w$ occurs in $T[1..i-\Pals[i]-1]$.

\subsubsection{Na\"ive quadratic-time algorithm with RMQs}

Let $T' = T\$\rev{T}\#$.
We construct the suffix array $\sa_{T'}$,
the reversed suffix array $\rsa_{T'}$,
and the LCP array $\lcp_{T'}$ for $T'$.

For each $\Pals[i]$ in $T$,
for each gap size $G = 1, \ldots, i-\Pals[i]-1$,
we compute $W = \lcplen(\rev{T[1..i-\Pals[i]-G]}, T[i+\Pals[i]+1..n])$
in $O(1)$ time by an RMQ on the LCP array $\lcp_{T'}$.
Then, the gap sizes $G$ with largest values of $W$ give
all longest SAGPs for pivot $i$.
Since we test $O(n)$ gap sizes for every pivot $i$,
it takes a total of $O(n^2)$ time to compute $\LongestSAGPone(T)$.
The working space of this method is $O(n)$.

\subsubsection{Simple quadratic-time algorithm based on suffix array}
Given a string $T$,
we construct $\sa_{T'}$, $\rsa_{T'}$, and $\lcp_{T'}$
for string $T' = T\$\rev{T}\#$ as in the previous subsection.
Further, for each position $n+2 \leq j \leq 2n+1$ 
in the reversed part $\rev{T}$ of $T' = T\$\rev{T}\#$,
let $\op(j)$ denote its ``original'' position in the string $T$,
namely, let $\op(j) = 2n-j+2$.
Let $e$ be any entry of $\sa_{T'}$ such that
$\sa_{T'}[e] \geq n+2$.
We associate each such entry of $\sa_{T'}$ with $\op(\sa_{T'}[e])$.

Let $\sa_{T'}[k] = i+\Pals[i]+1$, namely,
the $k$-th entry of $\sa_{T'}$ corresponds to the suffix
$T[i+\Pals[i]+1..n]$ of $T$.
Now, the task is to find the longest prefix $\rev{w}$ of
$T[i+\Pals[i]+1..n]$ such that $w$ occurs completely inside $T[1..i-\Pals[i]-1]$.
Let $b = i-\Pals[i]+1$, namely, $b$ is the beginning position of
the maximal palindrome $u\rev{u}$ centered at $i$.
We can find $w$ for any maximal SAGP $wgu\rev{u}\rev{w}$ for pivot $i$
by traversing $\sa_{T'}$ from the $k$-th entry forward and backward,
until we encounter the nearest entries $p < k$ and $q > k$
on $\sa_{T'}$ such that $\op(\sa_{T'}[p]) < b-1$ and $\op(\sa_{T'}[q]) < b-1$,
if they exist.
The size $W$ of $w$ is equal to
\begin{equation}
\max\{\min\{\lcp_{T'}[p+1], \ldots, \lcp_{T'}[k]\}, \min\{\lcp_{T'}[k+1], \ldots, \lcp_{T'}[q]\}\}.
\end{equation}
Assume w.l.o.g. that $p$ gives a larger lcp value with $k$,
i.e. $W = \min\{\lcp_{T'}[p+1], \ldots, \lcp_{T'}[k]\}$.
Let $s$ be the largest entry of $\sa_{T'}$ such that
$s < p$ and $\lcp_{T'}[s+1] < W$.
Then, any entry $t$ of $\sa_{T'}$
such that $s < t \leq p$ and $\op(\sa_{T'}[t]) < b-1$
corresponds to an occurrence of a longest SAGP for pivot $i$,
with gap size $b - \op(\sa_{T'}[t])-1$.
We output longest SAGP $(i, W, b - \op(\sa_{T'}[t])-1, |u|)$
for each such $t$.
The case where $q$ gives a larger lcp value with $k$,
or $p$ and $q$ give the same lcp values with $k$ can be
treated similarly.

We find $p$ and $s$ by simply traversing $\sa_{T'}$ from $k$.
Since the distance from $k$ to $s$ is $O(n)$,
the above algorithm takes $O(n^2)$ time.
The working space is $O(n)$.

\subsubsection{Algorithm based on suffix array and predecessor/successor queries}
Let $\occ_1 = |\LongestSAGPone(T)|$.
For any position $r$ in $T$,
we say that the entry $j$ of $\sa_{T'}$ is \emph{active} w.r.t. $r$
iff $\op(\sa_{T'}[j]) < r-1$.
Let $\Active(r)$ denote the set of active entries of $\sa_{T'}$
for position $r$, namely,
$\Active(r) = \{j \mid \op(\sa_{T'}[j]) < r-1\}$.

Let $t_1 = p$, and let $t_2, \ldots, t_{h}$ be the decreasing
sequence of entries of $\sa_{T'}$ which correspond to the
occurrences of longest SAGPs for pivot $i$.
Notice that for all $1 \leq \ell \leq h$
we have $\op(\sa_{T'}[t_\ell]) < b-1$ and hence
$t_{\ell} \in \Active(b)$, where $b = i-|u|+1$.
Then, finding $t_1$ reduces to a predecessor query 
for $k$ in $\Active(b)$.
Also, finding $t_{\ell}$ for $2 \leq \ell \leq h$ reduces
to a predecessor query for $t_{\ell-1}$ in $\Active(b)$.

To effectively use the above observation,
we compute an array $U$ of size $n$ from $Pals$ such that
$U[b]$ stores a list of all maximal palindromes in $T$
which begin at position $b$ if they exist,
and $U[b]$ is nil otherwise.
$U$ can be computed in $O(n)$ time e.g., by bucket sort.
After computing $U$,
we process $b = 1, \ldots, n$ in increasing order.
Assume that when we process a certain value of $b$,
we have maintained a dynamic predecessor/successor query data
structure for $\Active(b)$.
The key is that 
the same set $\Active(b)$ can be used to compute the longest SAGPs
for every element in $U[b]$,
and hence we can use the same predecessor/successor data structure
for all of them.
After processing all elements in $U[b]$,
we insert all elements of $\Active(b+1) \setminus \Active(b)$
to the predecessor/successor data structure.
Each element to insert can be easily found in constant time.

Since we perform $O(n + \occ_1)$ predecessor/successor queries
and $O(n)$ insertion operations in total,
we obtain the following theorem.

\begin{theorem} \label{theo:predecessor_algorithm}
  Given a string $T$ of size $n$ over
  an integer alphabet of size $\sigma = n^{O(1)}$,
  we can compute $\LongestSAGPone(T)$ 
  in $O(n(u(n,n)+q(n,n)) + \occ_1 \cdot q(n,n))$ time 
  with $O(n + s(n,n))$ space
  by using the suffix array and a predecessor/successor data structure,
  where $\occ_1 = |\LongestSAGPone(T)|$.
\end{theorem}

Since every element of $\Active(b)$ for any $b$ is in range $[1,2n+2]$,
we can employ the van Emde Boas tree~\cite{Boas75} as
the dynamic predecessor/successor data structure
using $O(n)$ total space.
Thus we obtain the following theorem.
\begin{theorem}
  Given a string $T$ of size $n$ over
  an integer alphabet of size $\sigma = n^{O(1)}$,
  we can compute $\LongestSAGPone(T)$ 
  in $O((n+\occ_1) \log \log n)$ time and $O(n)$ space
  by using the suffix array and the van Emde Boas tree,
  where $\occ_1 = |\LongestSAGPone(T)|$.
\end{theorem}

\begin{figure}[t]
	\centerline{
		\includegraphics[scale=0.65]{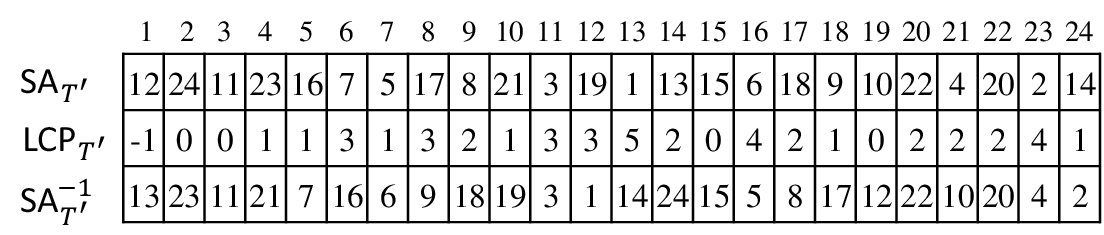}
	}
	\caption{$\sa_{T'}$, $\lcp_{T'}$ and $\rsa_{T'}$
		for string $T' = \mathtt{acacabaabca\$acbaabacaca\#}$.}
	\label{fig:sagp_array}
\end{figure}

\begin{example}
Let $T=\mathtt{acacabaabca}$ and $T' = \mathtt{acacabaabca}\$\mathtt{acbaabacaca}\#$, where $T'= T\$ \rev{T} \#$.
First, we compute $\Pals$ and the array $U$.
Assume we are now processing position $b=6$ in $T$, then $U[6] = \{(6, 9)\}$,
where $(6, 9)$ represents the maximal palindrome $T[6..9] = \mathtt{baab}$.
Thus we consider pivot $i = b + \lceil (9-6+1)/2 \rceil -1 = 7$.

First, we construct the suffix array $\sa_{T'}$,
the reversed suffix array $\rsa_{T'}$,
and the LCP array $\lcp_{T'}$ for $T'$.
Figure~\ref{fig:sagp_array} shows these arrays. 
Let $k$ be the integer such that $\sa_{T'}[k] = i + \lceil (9-6+1)/2 \rceil + 1 = 
7 + 3 = 10$, namely $k = 19$.
This can be obtained from $\rsa[10] = 19$ (see Figure~\ref{fig:sagp_array}).
To compute the longest $w$, we traverse $\sa_{T'}[19]$ forward and backward,
until we encounter the nearest entries $p < k$ and $q > k$ on $\sa_{T'}$
such that $\op(\sa_{T'}[p]) < 5$ and $\op(\sa_{T'}[q]) < 5$.
Note that these are equivalent to predecessor/successor queries for $19$,
respectively.
Then, we can find $p=10$ and $q=20$. 
Then, the size $W$ of $w$ is computed by
\[
W = \max\{\min\{\lcp_{T'}[11], \ldots, \lcp_{T'}[19]\}, \min\{\lcp_{T'}[20], \ldots, \lcp_{T'}[20]\}\},
\]
and we obtain $W=2$.
In this case, $q=20$ gives a larger lcp value with $k = 19$.
Thus, we output a canonical longest SAGP $(7, 2, 3, 2) = \mathtt{ac | \underline{aca} | ba | ab | ca}$.
We further traverse $\sa_{T'}$ from the $20$th entry backward as long as successive entries $s$ 
fulfill $\lcp_{T'}[s+1] \geq W$. 
Then, we find $s=22$, thus we output a canonical longest SAGP $(7, 2, 1, 2) = \mathtt{ac | \underline{a} | ba | ab | ca}$.
We further traverse $\sa_{T'}$ from the $17$th entry backward,
finally we reach the $24$th entry of $\sa_{T'}$,
which is the last entry of the suffix array.
Therefore, we finish the process for position $b = 6$. 
\end{example}

\subsubsection{Optimal-time algorithm based on suffix tree}

In this subsection, we show that
the problem can be solved in \emph{optimal} time and space,
using the following three suffix trees regarding the input string $T$.
Let $\mathcal{T}_1 = \suftree{T\$\rev{T}\#}$ for string $T\$\rev{T}\#$
of length $2n+2$,
and $\mathcal{T}_2 = \suftree{\rev{T}\#}$ of length $n+1$.
These suffix trees $\mathcal{T}_1$ and $\mathcal{T}_2$ are static, 
and thus can be constructed offline, in $O(n)$ time for an integer alphabet.
We also maintain a growing suffix tree 
$\mathcal{T}^\prime_2 = \suftree{\rev{T}[k..n])\#}$
for decreasing $k = n, \ldots, 1$.

\begin{lemma} \label{lem:growing_suffix_tree_linear_time}
  Given $\mathcal{T}_2 = \suftree{\rev{T}\#}$,
  we can maintain $\mathcal{T}^\prime_2 = \suftree{\rev{T}[k..n]\#}$
  for decreasing $k = n, \ldots, 1$ incrementally,
  in $O(n)$ total time for an integer alphabet of size $n^{O(1)}$.
\end{lemma}
\begin{figure}[t]
	\centerline{
		\includegraphics[scale=0.60, clip]{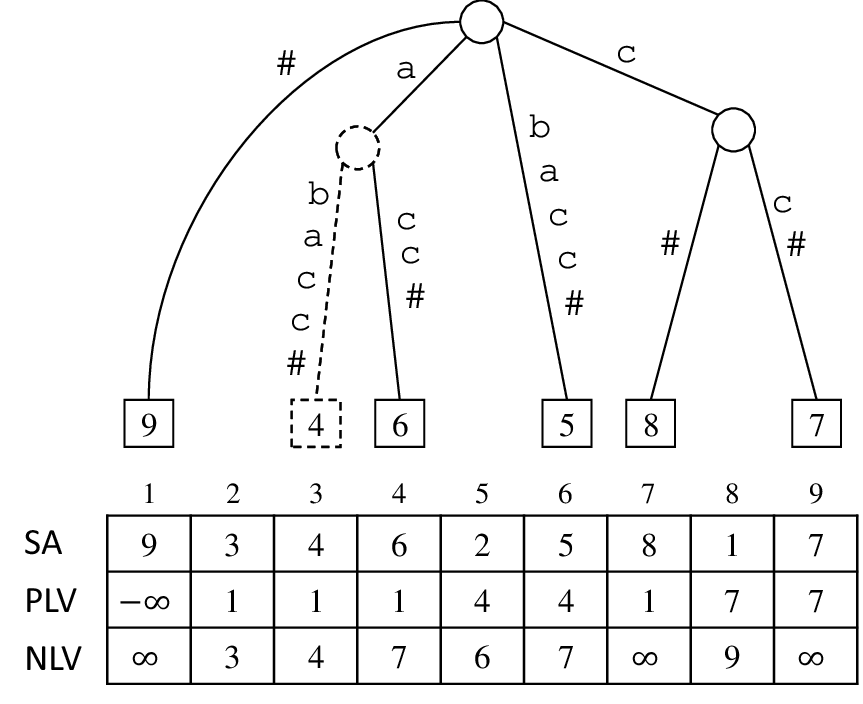}
	}
	\caption{
		Illustration for the proof of Lemma~\ref{lem:growing_suffix_tree_linear_time}.
		Consider string $T = \mathtt{ccabaabc}$.
		Here, we illustrate incremental construction of the growing suffix tree
		$\mathcal{T}'_2$ for its reversed string $\rev{T}\# = \mathtt{cbaabacc}\#$.
		Assume we have constructed
		$\mathcal{T}'_2 = \suftree{\rev{T}[5..8]\#} = \suftree{\mathtt{bacc}\#}$,
		and we are inserting a new leaf for the next suffix
		$\rev{T}[4..8]\# = \mathtt{abacc}\#$ starting at position
		$4$ in $\rev{T}\#$. $\sa^{-1}[4] = 3$, hence we focus on
		the suffixes of $\rev{T}\#$ starting at positions 
		$\sa[\PLV[3]] = \sa[1] = 9$ and $\sa[\NLV[3]] = \sa[4] = 6$.
		Since $\lcplen(\rev{T}[9..8]\#, \rev{T}[4..8]\#)
		= \lcplen(\#, \mathtt{abacc}\#) = 0
		< \lcplen(\rev{T}[4..8]\#, \rev{T}[6..8]\#)
		= \lcplen(\mathtt{abacc}\#, \mathtt{acc}\#) = 1$,
		we split the corresponding edge
		and insert the new leaf $\sa[3] = 4$ as the left neighbor of the leaf $6$.
	}
	
	\label{fig:increment_ST}
\end{figure}

\begin{proof}
        We build an LCA data structure on $\mathcal{T}_2 = \suftree{\rev{T}\#}$.
        We also construct a \emph{level ancestor} data structure~\cite{BenderF04}
        on $\mathcal{T}_2$ in linear preprocessing time and space,
        so that later, given a node $u$ and a positive integer $\ell$,
        the $\ell$th ancestor of $u$ in $\mathcal{T}_2$ can be answered in $O(1)$ time.
        Additionally, we use $\sa_{\rev{T}\#}$ and $\sa^{-1}_{\rev{T}\#}$ in our algorithm,
	and throughout this proof we abbreviate $\sa_{\rev{T}\#}$ as $\sa$
	and $\sa_{\rev{T}\#}^{-1}$ as $\sa^{-1}$ for simplicity.
	Let $\PLV$ and $\NLV$ be arrays of size $n+1$ each, such that
	for every $1 \leq j \leq n+1$,
	\begin{eqnarray*}
		\PLV[j] & = & \max(\{j' \mid 1 \leq j' < j, \sa[j'] > \sa[j]\} \cup \{- \infty\}), \\
		\NLV[j] & = & \min(\{j' \mid j < j' \leq n+1, \sa[j'] > \sa[j]\} \cup \{\infty\}).
	\end{eqnarray*}
	Intuitively, $\PLV[j]$ and $\NLV[j]$ indicate the entries of $\sa$
	that correspond to the lexicographically closest suffixes 
	to the left and to the right of the suffix $\rev{T}[\sa[j]..n]\#$
	which occur positions larger than $j$, respectively.
	If such entries do not exist, then let $\PLV[j] = -\infty$
	and $\NLV[j] = \infty$.
	See also Figure~\ref{fig:increment_ST} for
	concrete examples of $\PLV$ and $\NLV$ arrays.
	
	Suppose we have constructed $\mathcal{T}^\prime_2 = \suftree{\rev{T}[k+1..n]\#}$
	up to position $k+1$, 
	and we wish to update it with the new character $\rev{T}[k]$ at position $k$.
        We maintain an NMA data structure on the full suffix tree $\mathcal{T}_2$
        such that a node in $\mathcal{T}_2$ is marked iff  
        it exists in the growing suffix tree $\mathcal{T}^\prime_2$. 
        Namely, in each step of our algorithm,
        the induced tree with the marked nodes of $\mathcal{T}_2$
        coincides with $\mathcal{T}^\prime_2$. 
        
	Now, what is required here is to insert a new leaf 
	corresponding to the suffix $\rev{T}[k..n]\#$ to the growing suffix tree.
	If we use a variant of Weiner's algorithm~\cite{Weiner73},
	we can do this in $O(\log \sigma)$ (amortized) time,
	but this becomes $O(\log n)$ for integer alphabets of size $\sigma = n^{O(1)}$,
	and thus it is not enough for our goal.
	To achieve $O(1)$ update time per character, 
	we utilize the power of the full suffix tree 
	$\mathcal{T}_2 = \suftree{\rev{T}\#}$
	and four arrays $\sa$, $\sa^{-1}$, $\PLV$, and $\NLV$.
	
	In what follows,
	we focus on the case where $\PLV[\sa^{-1}[k]] \neq -\infty$
	and $\NLV[\sa^{-1}[k]] \neq \infty$.
	The case where $\PLV[\sa^{-1}[k]] = \infty$ or $\NLV[\sa^{-1}[k]]
	= \infty$
	is simpler and can be treated similarly.
	A key observation is that 
	there is a one-to-one correspondence between 
	every leaf of $\suftree{\rev{T}[k+1..n]\#}$
	and an entry of $\sa$ which stores 
	a position in $\rev{T}\#$ which is \emph{larger} than $k$.
	Hence, $\sa[\PLV[\sa^{-1}[k]]]$ and $\sa[\NLV[\sa^{-1}[k]]]$ will be,
	respectively,
	the left and right neighboring leaves of the new leaf $k$ 
	in the updated suffix tree $\suftree{\rev{T}[k..n]\#}$.
	
	Given the new position $k$, 
	we compute the following values $L$ and $R$:
	\begin{eqnarray*}
		L & = & \lcplen((\rev{T}\#)[\sa[\PLV[\sa^{-1}[k]]..n+1]], (\rev{T}\#)[k..n+1]), \\
		R & = & \lcplen((\rev{T}\#)[k..n+1], (\rev{T}\#)[\sa[\NLV[\sa^{-1}[k]]..n+1]]).
	\end{eqnarray*}
        Let $v_L$ and $v_R$ be the LCA nodes in the full suffix tree $\mathcal{T}_2$
        that correspond to the LCP values $L$ and $R$, respectively
        (i.e., the string depths of $v_L$ and $v_R$ are respectively $L$ and $R$).
        Both $v_L$ and $v_R$ can be found in $O(1)$ time 
        by LCA queries on the full suffix tree $\mathcal{T}_2$.
	Depending on the values of $L$ and $R$, we have the following cases.
	\begin{itemize}
		\item If $L \geq R$, then leaf $\sa[\PLV[\sa^{-1}[k]]]$ will be the left neighbor
		  of the new leaf $k$ in the updated suffix tree.
                  We then find the parent $\parent(v_L)$ of $v_L$ in the growing tree $\mathcal{T}_2^\prime$,
                  by an NMA query from $v_L$ on the full suffix tree $\mathcal{T}_2$.
                  Then, in the growing tree,
                  we split the corresponding out-going edge of $\parent(v_L)$ 
                  with the string depth of $v_L$ if necessary,
                  and insert a new leaf $k$.
                  In case where the parent of this leaf $k$ was newly created
                  after splitting the edge,
                  then we mark the corresponding node in the full suffix tree $\mathcal{T}_2$.

                  What remains is how to efficiently locate this corresponding out-going edge of $\parent(v_L)$ to be split.
                  For this sake, we associate each edge $f$ of $\mathcal{T}_2^\prime$
                  with the first edge of the path in $\mathcal{T}_2$
                  that corresponds to $f$.
                  We also precompute the depths of all nodes in the full suffix tree
                  $\mathcal{T}_2$ by a standard linear-time tree traversal.
                  Let $u$ be the (marked) child of $\parent(v_L)$ in the growing tree $\mathcal{T}_2^\prime$
                  such that
                  the edge from $\parent(v_L)$ to $u$ is split and $v_L$ is newly inserted in between.
                  Let $x$ be the label of the edge from $\parent(v_L)$ to $u$ in $\mathcal{T}_2^\prime$,
                  and let $e$ be the first edge of the corresponding path of $\mathcal{T}_2$
                  that spells out $x$.
                  Since we know the depth of $\parent(v_L)$ in the full suffix tree $\mathcal{T}_2$,
                  we can find $e$ in $O(1)$ time by a level ancestor query
                  from leaf $\sa[\PLV[\sa^{-1}[k]]]$ on the full suffix tree $\mathcal{T}_2$.
                  Thus we can locate the edge from $\parent(v_L)$ to $u$ using $e$, in constant time.
                  Let $x = yz$ be the partition of $x$ such that
                  the path from $\parent(v_L)$ to $v_L$ spells out $y$
                  and the path from $v_L$ to $u$ spells out $z$.
                  Let $e'$ be the out-going edge of $v_L$ in $\mathcal{T}_2$, that is the first edge in this path spelling out $z$.
                  After the update to the growing suffix tree, the new edge from $\parent(v_L)$ to $v_L$ is associated to $e$,
                  and the new edge from $v_L$ to $u$ is associated to $e'$.
                  The edge $e'$ can also be found by a level ancestor query from $\sa[\PLV[\sa^{-1}[k]]]$ on $\mathcal{T}_2$.
                  See Figure~\ref{fig:growing_tree} for illustration.
                  
		\item If $L < R$, then leaf $\sa[\NLV[\sa^{-1}[k]]]$ will be the right neighbor
		  of leaf $k$ in the updated suffix tree.
                  This case can be treated in a similar manner as the afore-mentioned case.
		
	\end{itemize}
	We then associate the new leaf $k$ with the $\sa^{-1}[k]$-th entry of $\sa$
	so that later, given $k$, we can access to this leaf on $\sa$
	in $O(1)$ time.
	See also Figure~\ref{fig:increment_ST} for a concrete example
	on how we insert a new leaf to the growing suffix tree.

	Let us analyze the efficiency of our algorithm.
	Given $\sa$, $\PLV$ and $\NLV$ can be constructed in 
	$O(n)$ time~\cite{CrochemoreIIKRW13}.
	Then, given a position $k$ in string $\rev{T}\#$, 
	we can access the leaves 
	$\sa[\PLV[\sa^{-1}[k]]]$ and $\sa[\NLV[\sa^{-1}[k]]]$ of the full suffix tree 
	$\mathcal{T}_2 = \suftree{\rev{T}\#}$
	in $O(1)$ time using $\sa$, $\sa^{-1}$, $\PLV$, and $\NLV$ arrays.
	The values of $L$ and $R$ can be computed in $O(1)$ time
	by two LCA queries on the full suffix tree $\mathcal{T}_2$.
	In each of the afore-mentioned cases,
        we perform at most two level ancestor queries on $\mathcal{T}_2$,
        using $O(1)$ time each.
	Thus it takes $O(1)$ time to insert a new leaf.
        Thus, it takes $O(1)$ time to insert a new leaf to 
	the growing suffix tree $\mathcal{T}'_2$.
	This completes the proof.
\end{proof}

\begin{figure}[t]
	\centerline{
	  \includegraphics[scale=0.4]{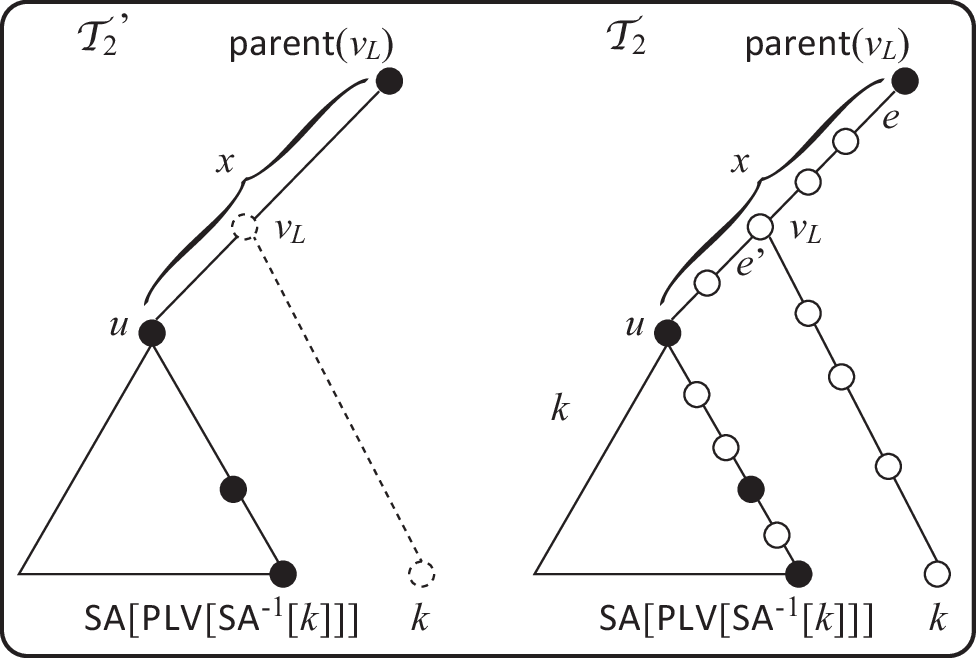}
          \hfill
  	  \includegraphics[scale=0.4]{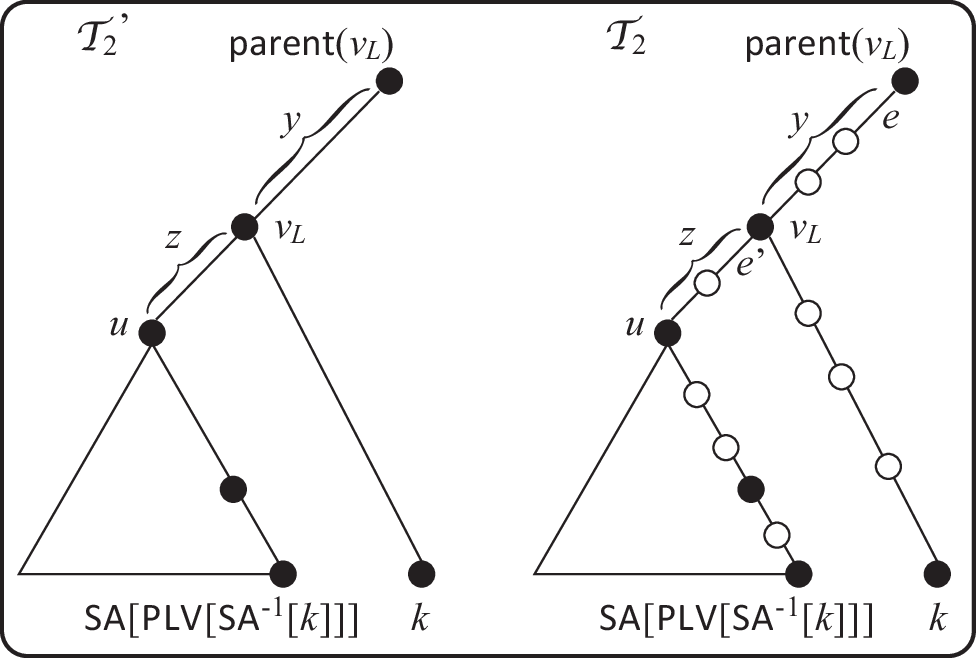}
	}
	\caption{Left: The growing suffix tree $\mathcal{T}_2^\prime$ and the full suffix tree $\mathcal{T}_2$ before the update. Black nodes are marked nodes, and white ones are unmarked nodes. We find $v_L$ by an LCA query from the leaves for $\sa[\PLV[\sa^{-1}[k]]]$ and $k$ on $\mathcal{T}_2$, and find $\parent(v_L)$ by an NMA query from $v_L$ on $\mathcal{T}_2$. We then find the edge $e$ on $\mathcal{T}_2$ by a level ancestor query from $\sa[\PLV[\sa^{-1}[k]]]$ on $\mathcal{T}_2$, and locate the edge from $\parent(v_L)$ to $u$ on the growing tree $\mathcal{T}_2^\prime$ using a link from $e$. Right: The growing suffix tree $\mathcal{T}_2^\prime$ and the full suffix tree $\mathcal{T}_2$ after the update. Now the edge from $\parent(v_L)$ to $v_L$ on $\mathcal{T}_2^\prime$ is associated to $e$, and the edge from $v_L$ to $u$ on $\mathcal{T}_2^\prime$ is associated to $e'$.}	
	\label{fig:growing_tree}
\end{figure}

We now obtain the main result for this subsection:
\begin{theorem} \label{theo:maximal_suffix_tree_linear}
  Given a string $T$ of length $n$ over
  an integer alphabet of size $\sigma = n^{O(1)}$,
  we can compute $\LongestSAGPone(T)$  
  in optimal $O(n+\occ_1)$ time and $O(n)$ space
  by using suffix trees,
  where $\occ_1 = |\LongestSAGPone(T)|$.
\end{theorem}

\begin{proof}
  We first compute the array $U$.
  Consider an arbitrary fixed $b$,
  and let $u\rev{u}$ be a maximal palindrome stored in $U[b]$
  whose center is $i = b+|u|-1$.
  Assume that we have a growing suffix tree $\mathcal{T}'_2$
  for string $\rev{T}[n-b+1..n]\#$
  which corresponds to the prefix $T[1..b]$ of $T$ of size $b$.
  We use a similar strategy as the suffix array based algorithms.
  For each position $2n-b+2 \leq j \leq 2n+1$ in string $T' = T\$\rev{T}\#$,
  $1 \leq \op(j) \leq b-2$.
  We maintain the NMA data structure over the suffix tree $\mathcal{T}_1$
  for string $T'$ so that all the ancestors of
  the leaves whose corresponding suffixes start
  at positions $2n-b+2 \leq j \leq 2n+1$ are marked,
  and any other nodes in $\mathcal{T}_1$ remain unmarked at this step.

  As in the suffix-array based algorithms,
  the task is to find the longest prefix $\rev{w}$ of
  $T[i+|\rev{u}|+1..n]$ such that $w$ occurs completely inside 
  $T[1..b-2] = T[1..i-|u|-1]$.
  In so doing, we perform an NMA query from the leaf 
  $i+|\rev{u}|+1$ of $\mathcal{T}_1$,
  and let $v$ be the answer to the NMA query.
  By the way how we have maintained the NMA data structure, 
  it follows that $\Str(v) = \rev{w}$.

  To obtain the occurrences of $w$ in $T[1..b-2]$,
  we switch to $\mathcal{T}'_2$, and traverse the subtree rooted at $v$.
  Then, for any leaf $\ell$ in the subtree,
  $(i, |\Str(v)|, b-\op(\ell), |u|)$ is a canonical longest SAGP for pivot $i$
  (see also Figure~\ref{fig:sagp_ex_stree}).

  After processing all the maximal palindromes in $U[b]$,
  we mark all unmarked ancestors of the leaf $2n-b$ of $\mathcal{T}_1$
  in a bottom-up manner, until we encounter the lowest
  ancestor that is already marked.
  This operation is a preprocessing for the maximal palindromes in $U[b+1]$,
  as we will be interested in 
  the positions between $1$ and $\op(2n-b) = b-1$ in $T$.
  In this preprocessing, each unmarked node is marked at most once,
  and each marked node will remain marked.
  In addition, we update the growing suffix tree $\mathcal{T}'_2$
  by inserting the new leaf for $\rev{T}[n-b..n]\#$.
  
  We analyze the time complexity of this algorithm.
  Since all maximal palindromes in $U[b]$ begin at position $b$ in $T$,
  we can use the same set of marked nodes on $\mathcal{T}_1$ for all of 
  those in $U[b]$.
  Thus, the total cost to update the NMA data structure for all $b$'s
  is linear in the number of unmarked nodes that later become marked,
  which is $O(n)$ overall.
  The cost for traversing the subtree of $\mathcal{T}'_2$
  to find the occurrences of $w$ can be charged to
  the number of canonical longest SAGPs to output for each pivot,
  thus it takes $O(\occ_1)$ time for all pivots.
  Updating the growing suffix tree $\mathcal{T}'_2$ takes 
  overall $O(n)$ time by Lemma~\ref{lem:growing_suffix_tree_linear_time}.
  What remains is how to efficiently link the new internal node introduced
  in the growing suffix tree $\mathcal{T}'_2$, 
  to its corresponding node in the static suffix tree 
  $\mathcal{T}_1$ for string $T'$.
  This can be done in $O(1)$ time using a similar technique
  based on LCA queries on $\mathcal{T}_1$,
  as in the proof of Lemma~\ref{lem:growing_suffix_tree_linear_time}.
  Summing up all the above costs, we obtain 
  $O(n + \occ_1)$ optimal running time and $O(n)$ working space.
\end{proof}

\begin{example}
Let $T=\mathtt{acacabaabca}$ and $T' = \mathtt{acacabaabca}\$\mathtt{acbaabacaca}\#$, where $T'= T\$ \rev{T} \#$.
First, we compute $\Pals$ and the array $U$.
Assume we are now processing position $b=6$ in $T$, then $U[6] = \{(6, 9)\}$,
where $(6, 9)$ represents the maximal palindrome $T[6..9] = \mathtt{baab}$.
Thus we consider pivot $i = b + \lceil (9-6+1)/2 \rceil -1 = 7$.

First, we construct the suffix tree $\mathcal{T}_1 = \suftree{T\$\rev{T}\#}$.
Suppose that we have constructed $\mathcal{T}'_2=\suftree{\rev{T}[8..11]\#}$ and
marked all ancestors of every leaf $v$ such that $19 < v \leq 24$ in $\mathcal{T}_1$.
In Figure~\ref{fig:sagp_ex_stree}, 
we show interesting parts of $\mathcal{T}_1$ and $\mathcal{T}'_2$.

To compute the longest $w$, we perform an NMA query from the leaf 
$i+|\rev{u}| +1 = 10$ of $\mathcal{T}_1$.
As can be seen in Figure~\ref{fig:sagp_ex_stree}, 
we obtain the nearest marked node $v = \nma{\mathcal{T}_1}{10}$.
Thus, we know that $\rev{w} = \mathtt{ca}$.
Next, we switch from the node $v$ of $\mathcal{T}_1$ to its 
corresponding node $v'$ of $\mathcal{T}'_2$ using a link between them.
Then, we traverse the subtree rooted at $v'$ 
and obtain all occurrences of $\rev{w}$, namely
$\rev{w} = \rev{T}[10..11] = \rev{T}[8..9] = \mathtt{ca}$
at positions $10$ and $8$ in the reversed string $\rev{T}\#$.
Since $\op(10) = 2$ and $\op(8) = 4$,
we obtain the canonical longest SAGPs 
$(7, 2, 3, 2) = \mathtt{ac | \underline{aca} | ba | ab | ca}$,
and $(7, 2, 1, 2) = \mathtt{ac | \underline{a} | ba | ab | ca}$
for pivot $7$.

\begin{figure}[t!]
	\centerline{
		\includegraphics[scale=0.30,clip]{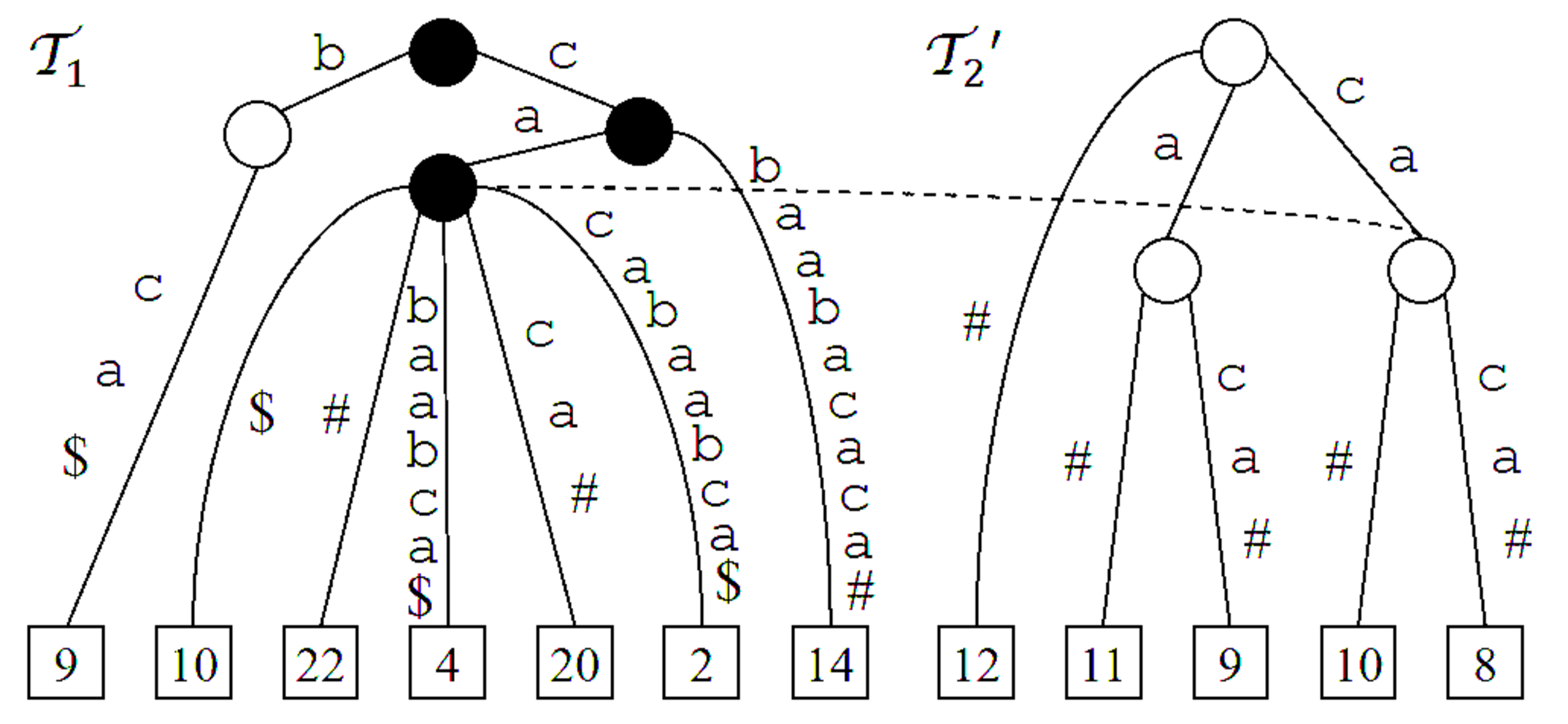}
	}  
	\caption{
		Showing interesting parts of $\mathcal{T}_1 = \suftree{T'}$
		and $\mathcal{T}'_2 = \suftree{\rev{T}[8..11]\#}$,
		where $T = \mathtt{acacabaabca}$, $\rev{T} = \mathtt{acbaabacaca}$ and    
		$T' = \mathtt{acacabaabca\$acbaabacaca\#}$.
		In $\mathcal{T}_1$, we represent the marked internal nodes by black circles,
		the unmarked internal nodes by white circles,
		and the leaves by squares in which the numbers denote
		the beginning positions of the corresponding suffixes in the string.
		The dotted line represents the link between 
		the node for string $\mathtt{c}$ in $\mathcal{T}_1$ and 
		that in $\mathcal{T}'_2$.
	}
	\label{fig:sagp_ex_stree}
\end{figure}

\end{example}

\subsection{Computing $\LongestSAGPtwo(T)$ for type-2 positions} \label{sec:type-2}

In this subsection,
we present an algorithm to compute $\LongestSAGPtwo(T)$ in a given string $T$,
corresponding to the line~\ref{alg:compute_SAGP:line:computeTwo} in Algorithm~\ref{alg:compute_SAGP}.

\begin{lemma} \label{lem:type-2_position}
  Every (not necessarily longest) 
  SAGP for type-2 pivot $i$ must end at one of the positions
  between $i+2$ and $i+\Pals[i]$.
\end{lemma}

\begin{figure}[!t]
	\centerline{
		\includegraphics[scale=0.33, clip]{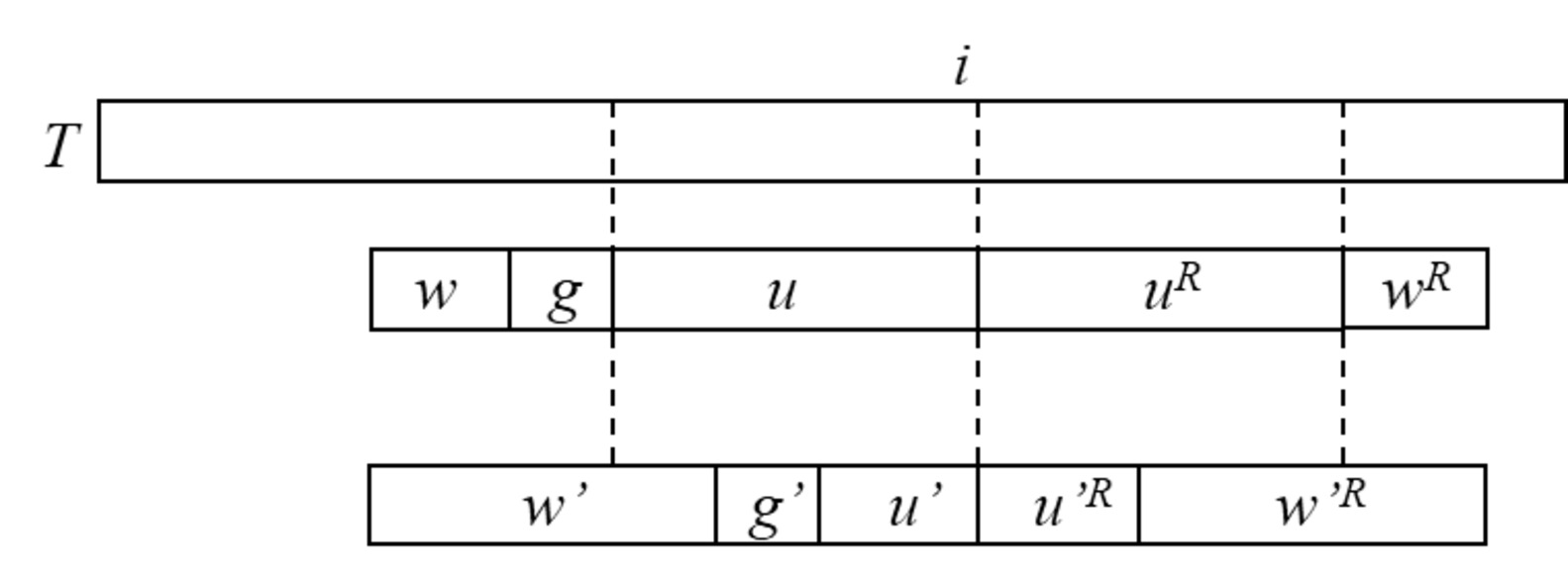}
	}
	\caption{Illustration for Lemma~\ref{lem:type-2_position}.}
	\label{fig:type-2_position}
\end{figure}

\begin{proof}
	See Figure~\ref{fig:type-2_position}.
	By definition, it is clear that
	any SAGP for pivot $i$
	must end at position $i+2$ or after that.
	Now, assume on the contrary that there exists a SAGP
	$w'g'u'\rev{u'}\rev{w'}$ for pivot $i$
	such that $i+|\rev{u'}\rev{w'}| > i+|\rev{u}|$
	(it ends after position $i+|\rev{u}|$),
	where $uu^R$ is the maximal palindrome centered at position $i$.
	Recall that since $i$ is a type-2 position,
	we have $|u'| < |u|$.
	Let $\rev{w}$ be the suffix of $\rev{w'}$
	of size $|\rev{u'}\rev{w'}|-|\rev{u}|$.
	Then, there exists a SAGP $wgu\rev{u}\rev{w}$ for pivot $i$
	where $|g| = |g'|$ and $u\rev{u}$ is the maximal palindrome centered at $i$.
	However, this contradicts that $i$ is a type-2 position.
	Hence, any SAGP for pivot $i$ must end at position
	$i+|\rev{u}|$ or before that.	
\end{proof}

\begin{lemma} \label{lem:wsize_of_type-2_position}
	For any type-2 position $i$ in string $T$, 
	if $w g u \rev{u} \rev{w}$ is a canonical longest SAGP for pivot $i$,
	then $|w|=1$.
\end{lemma}
\begin{figure}[t]
	\centerline{
		\includegraphics[scale=0.33, clip]{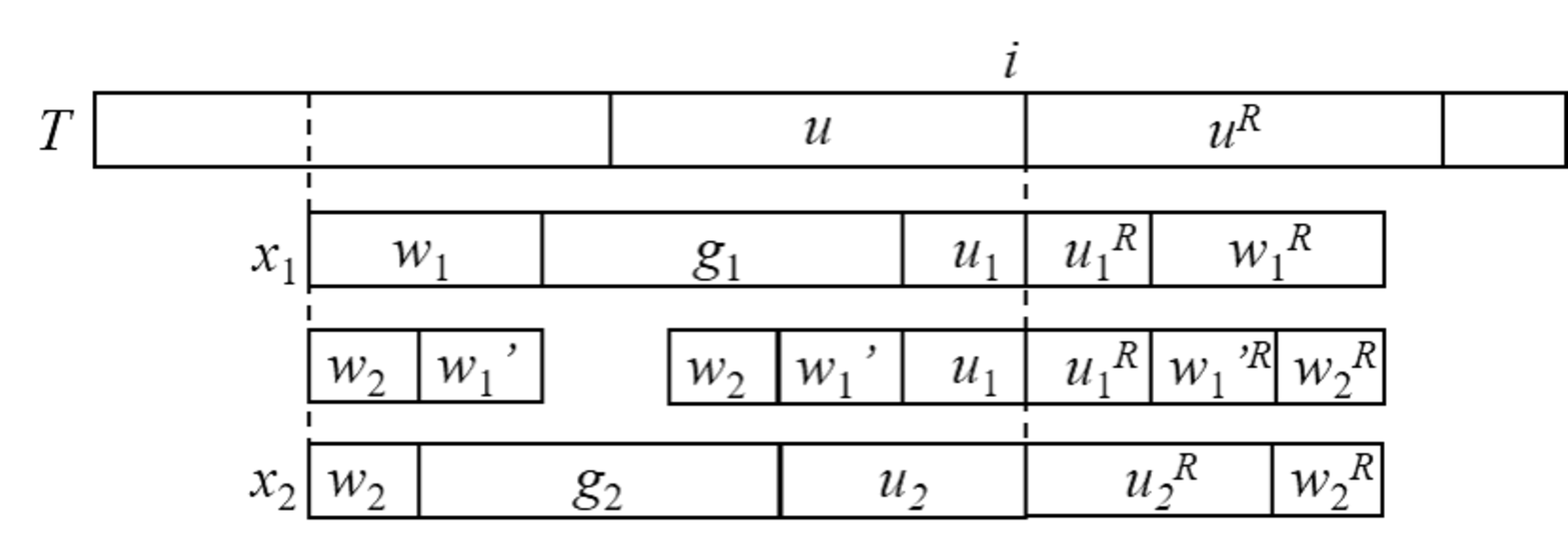}
	}  
	\caption{Illustration for Lemma~\ref{lem:wsize_of_type-2_position}.}
	\label{fig:type-2_w_size}
\end{figure}

\begin{proof}
	Let $x_1 = w_1g_1u_1\rev{u_1}\rev{w_1}$
	be a canonical longest SAGP for pivot $i$, and 
	on the contrary, suppose that $|w_1| \geq 2$. See also Fig~\ref{fig:type-2_w_size}.
	Then we can rewrite $w_1 = w_2 w_1'$ for two non-empty strings $w_2$ and $w_1'$.
	Let $u\rev{u}$ be the maximal palindrome centered at $i$.
	Since the position $i$ is type-2, 
	$\rev{u_1}\rev{w_1}$ is a prefix of $\rev{u}$ by Lemma~\ref{lem:type-2_position}, so that $w_1 u_1$ is a suffix of $u$.
	Moreover, let $u_2=w_1' u_1$ and $g_2$ be a string satisfying $g_2w_1' =w_1'g_1$.
	Then 
	$x_2 = w_2 g_2 \underline{u_2} \rev{u_2} \rev{w_2} 
	= w_2 g_2 w_1' u_1 \underline{\rev{u_2}} \rev{w_2} 
	= w_2 \underline{g_2 w_1'} u_1 \rev{ u_1} \rev{w_1'} \rev{w_2} 
	= \underline{w_2 w_1'} g_1 u_1 \rev{ u_1} \rev{w_1'} \rev{w_2}
	= w_1 g_1 u_1 \rev{u_1}\underline{\rev{w_1'} \rev{w_2}}$ 
	$= w_1 g_1 u_1 \rev{u_1}\rev{w_1} = x_1$, 
	that shows $x_2$ is also a SAGP for pivot $i$.
	Because $\armlen(x_2)= |w_2 u_2| = |w_1 u_1| = \armlen(x_1)$, $x_2$ is also a \emph{longest} SAGP for pivot $i$. Because $u_2 = w_1'u_1$ and $w_1' \neq \varepsilon$, we have $|u_2| < |u_1|$, which contradicts that $x_1$ is a \emph{canonical} longest SAGP for pivot $i$.
\end{proof}

For every type-2 position $i$ in $T$,
let $u=T[i..i+\Pals[i]]$.
By Lemma~\ref{lem:wsize_of_type-2_position}, any canonical longest SAGP is of the form $cgu\rev{u}c$ for $c\in \Sigma$.
For each $2 \leq k \leq \Pals[i]$, let $c_k = \rev{u}[k]$, and
let $\rev{u_k}$ be the proper prefix of $\rev{u}$ of length $k-1$.
Now, observe that the largest value of $k$ for which $\LeftMost[c_k] \leq i-|u_k|-1$
corresponds to a canonical longest SAGP for pivot $i$,
namely, $c_k g_k u_k \rev{u_k} c_k$ is a canonical longest
SAGP for pivot $i$, where $g_k = T[\LeftMost[c_k]+1..i-|u_k|]$.
In order to efficiently find the largest value of such, we consider a function $\FindR(t,i)$ defined by
\[
\FindR(t, i) = \min\{r \mid t \leq r < i, T[l] = T[r] \mbox{ for } 1 \leq l < r  \} \cup \{+\infty \}.
\]
\begin{lemma} \label{lem:FindR is useful}
	For any type-2 position $i$ in $T$,
	quadruple $(i,1,r - \LeftMost[T[r]], i - r)$ represents a canonical longest SAGP for pivot $i$, 
	where $r = \FindR(i - \Pals[i] + 1, i) \neq \infty$. 
	Moreover, its gap is the longest among all the canonical longest SAGPs for pivot $i$.
\end{lemma}

\begin{figure}[t]
	\centerline{
		\includegraphics[scale=0.33, clip]{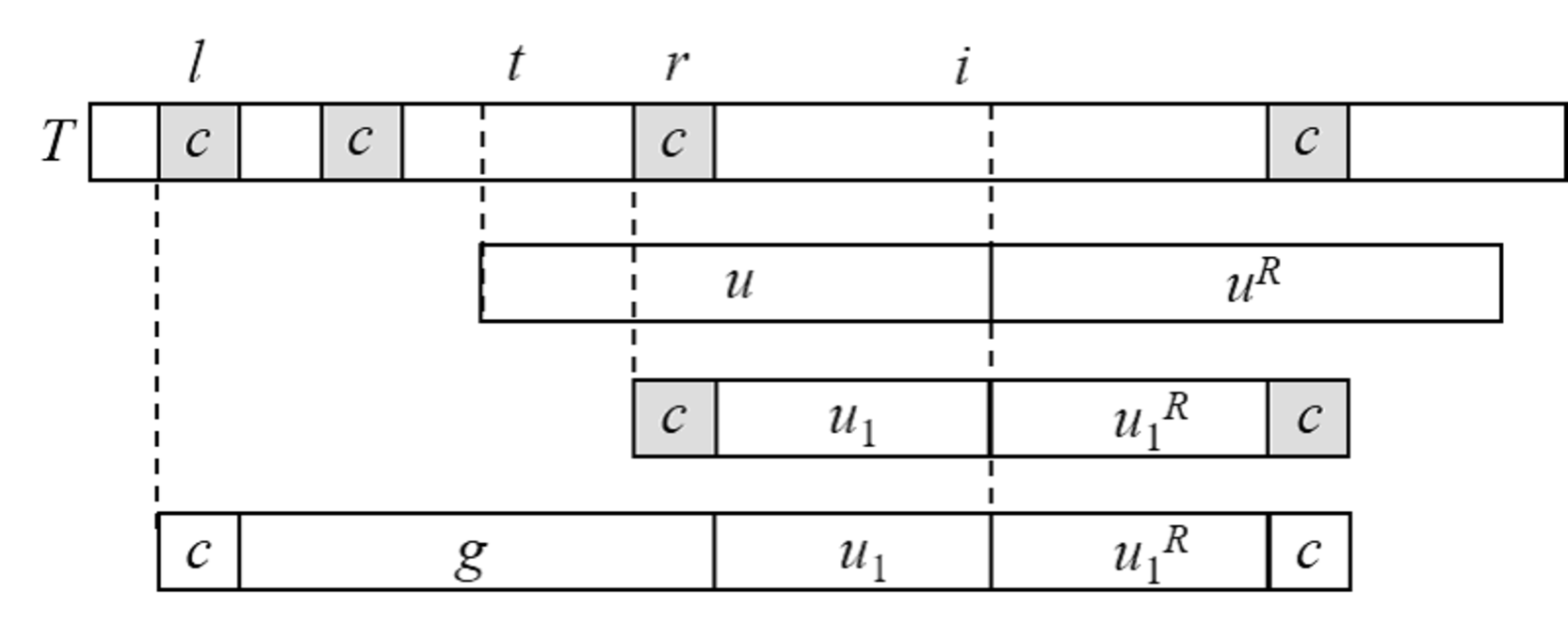}
	}  
	\caption{Illustration for Lemma~\ref{lem:FindR is useful}.}
	\label{fig:type-2_cano_longest}
\end{figure}

\begin{proof}
	See Figure~\ref{fig:type-2_cano_longest}.
	Let $t$ be the beginning position of $u$ in $T$, namely, $t = i-|u|+1$.
	Let $r = \FindR(t, i) < \infty$, and let $c = T[r]$.
	Then by definition of $\FindR(t, i)$, there exists $1 \leq l < r$ satisfying $T[l] = c$.
	Therefore, $x = (i, 1, r-l, i-r)$ is a SAGP for pivot $i$. 
	Moreover, $x$ is \emph{canonical longest} SAGP because $r$ is minimized, so that $|u_1|=i-r$ is maximized while $|w|$ is always 1.
	Recall that $\LeftMost[c]$ is the leftmost position $l$ satisfying $T[l] = c$.
	Hence, the gap size of the canonical longest SAGP $(i,1,r - \LeftMost[T[r]], i - r)$ is the longest.
\end{proof}

By Lemma~\ref{lem:FindR is useful}, we can compute a canonical longest SAGP for any type-2 pivot $i$ in $O(1)$ time,
assuming that the function $\FindR(t, i)$ returns a value in $O(1)$ time.
We define an array $\FindRarray$ of size $n$ by
\begin{align} \label{eqn:FindRarray}
\FindRarray[t] = \min\{r \mid t \leq r, \ \ T[l] = T[r] \mbox{ for } 1 \leq l < r  \} \cup \{+\infty \},
\end{align}
for $1 \leq t \leq n$. 
If the array $\FindRarray$ has already been computed, then $\FindR(t,i)$ can be obtained in $O(1)$ time by
$\FindR(t,i) = \FindRarray[t]$ if $\FindRarray[t] < i$, and $+\infty$ otherwise.

Algorithm~\ref{alg:compute FindRarray} shows a pseudo-code to compute $\FindRarray$.
Table~\ref{table:example_of_FindRarray} shows an example.

%

 \begin{algorithm2e}[t]
	\caption{constructing the array $\FindRarray$}
	\label{alg:compute FindRarray}
	\KwIn{string $T$ of length $n$}
	\KwOut{array $\FindRarray$ of size $n$}
	$\mathit{minpos} = +\infty$\;
	\For{$i=n$ {\bf downto} $1$}{
		\lIf{$\LeftMost[T[i]] \le i$}{$\mathit{minpos} = i$}
		$\FindRarray[i] = \mathit{minpos}$\;
	}
\end{algorithm2e}

\begin{table}[t]
	\caption{Arrays $\LeftMost$, $\NextPos$, and $\FindRarray$ for a string $T = {\tt dbbaacbcbad}$.
		For the sake of understanding, we also provide the values of $\minLout$ and $\minLin$ in the $i$-th loop of Algorithm~\ref{alg:compute FindRarray}.
		These values are computed from right to left.}
	\vspace{0.4cm}
	\label{table:example_of_FindRarray}
	\centering
	\begin{tabular}{|c|*{20}{>{\centering\arraybackslash}p{2.5em}|}}
		\hline
		& $\LeftMost$ \\ \hline
		\tt a & 4   \\
		\tt b & 2   \\
		\tt c & 6   \\
		\tt d & 1   \\ \hline
	\end{tabular}
	\quad
	\begin{tabular}{|c|*{20}{>{\centering\arraybackslash}p{0.8em}|}}
		\hline
		& 1 & 2 & 3	& 4	& 5	& 6	& 7	& 8	& 9	& 10 & 11 \\ \hline
		$T$		&\tt d	&\tt b	&\tt b &\tt a &\tt a &\tt c &\tt b &\tt c &\tt b &\tt a	&\tt d \\ \hline
		$\NextPos$ & 11 & 3 & 7 & 5 & 10 & 8 & 9  & $\infty$ & $\infty$ & $\infty$ & $\infty$ \\ \hline
		$\FindRarray$		& 3	& 3	& 3	& 5	& 5 & 7 & 7 & 8 & 9 & 10 & 11 \\ \hline
	\end{tabular}
\end{table}

\begin{lemma} \label{lem:correctness of FindRarray}
	Algorithm~\ref{alg:compute FindRarray} correctly computes the array $\FindRarray$ in $O(n)$ time and space.
\end{lemma}
	
	\begin{proof}
		The correctness of the computation of $\LeftMost$ is obvious.
		For each $i$, we maintaining invariant $\mathit{minpos} = \min\{r\mid i \le t, T[l] = T[r] \mbox{ for some } 1 \le l < r \}$.
		$\mathit{minpos}$ can be maintained by updating it if $i$ is not the leftmost occurrence position of $T[i]$.
		Since the definition of $\FindRarray[i]$ is the same as $\mathit{minpos}$ for each $i$,
		$\FindRarray[i]$ can be computed by assigning $\mathit{minpos}$ to $\FindRarray[i]$.
	\end{proof}

	By Lemma~\ref{lem:correctness of FindRarray},
	we can compute $\LongestSAGPtwo(T)$ for type-2 positions as follows.
	
	\begin{theorem} \label{theo:type-2_algorithm_2}
	  Given a string $T$ of length $n$ over an integer alphabet of size $n^{O(1)}$,
	  we can compute $\LongestSAGPtwo(T)$ 
	  in $O(n +\occ_2)$ time and $O(n)$ space,
	  where $\occ_2 = |\LongestSAGPtwo(T)|$.  
	\end{theorem}

	\begin{proof}
	For a given $T$, we first compute the array $\FindRarray$ by Algorithm~\ref{alg:compute FindRarray}.
	The correctness of the computation of $\NextPos$ is obvious.
	By Lemma~\ref{lem:FindR is useful}, we can get a canonical longest SAGP
	$x_1 = (i,1,|g_1|, \Pals[i]-r)$ 
	if $\Pals[i]-r\geq 1$, 
	in $O(1)$ time
	by referring to $\LeftMost$ and $\FindRarray$.
	Note that $x_1$ is the one whose gap $|g_1|$ is the longest.
	Let $b_1 = i - \Pals[i] + r - |g_1|$
	be the beginning position of $x_1$ in $T$.
	Then the next shorter canonical longest SAGP for the same pivot $i$ begins at position $b_2 = \NextPos[b_1]$.
	By repeating this process $b_{j+1} = \NextPos[b_j]$ while the gap size 
	$|g_j| = i - \Pals[i] + r - b_j$ is positive,
	we obtain all the canonical longest SAGPs for pivot $i$.
	Overall, we can compute all canonical longest SAGPs for all pivots in $T$ in $O(n + \occ_2)$ time.
	The space requirement is clearly $O(n)$.
	\end{proof}	
			
	We now have the main theorem from
	Theorem~\ref{theo:maximal_suffix_tree_linear},
	Lemma~\ref{lem:determine_types},
	Lemma~\ref{lem:correctness_alg:compute_SAGP},
	and Theorem~\ref{theo:type-2_algorithm_2} as follows.

	\begin{theorem}
	Given a string $T$ of length $n$ over an integer alphabet of size $n^{O(1)}$,
	Algorithm~\ref{alg:compute_SAGP} can compute $\LongestSAGP(T)$ 
	  	in optimal $O(n + \occ)$ time and $O(n)$ space,
	 	 where $\occ = |\LongestSAGP(T)|$.  
	\end{theorem}


\section{Experiments}

In this section, we show some experimental results which compare
performance of our algorithms for computing $\LongestSAGPone(T)$. 
We implemented the na\"ive quadratic-time algorithm (\textsf{Na\"ive}), 
the simple quadratic-time algorithm which traverses suffix arrays (\textsf{Traverse}),
and three versions of the algorithm based on suffix array and predecessor/successor data structure, each employing
red-black trees (\textsf{RB tree}),
Y-fast tries (\textsf{Y-fast trie}),
and van Emde Boas trees\footnote{We modified a van Emde Boas tree implementation from \url{https://code.google.com/archive/p/libveb/} so it works with Visual C++.} (\textsf{vEB tree}),
as the predecessor/successor data structure.

We implemented all these algorithms with Visual C++ 12.0 (2013),
and performed all experiments on 
a PC (Intel\copyright~Xeon CPU W3565 3.2GHz, 12GB of memory)
running on Windows 7 Professional.
In each problem, we generated a string randomly and got the average time for ten times attempts.

\begin{table}[t]
 \centering{
    \caption{Running times (in milli-sec.) on randomly generated strings of length $10000$, $50000$, and $100000$ with $|\Sigma|=10$.}
    \vspace*{0.3cm}
    \begin{tabular}{|r|r|r|r|r|r|} \hline
      \multicolumn{1}{|c|}{$n$} & \multicolumn{1}{|c|}{\textsf{Na\"ive}} 
      & \multicolumn{1}{|c|}{\textsf{Traverse}} & \multicolumn{1}{|c|}{\textsf{RB tree}} 
      & \multicolumn{1}{|c|}{\textsf{vEB tree}} & \multicolumn{1}{|c|}{\textsf{Y-fast trie}}  \\  \hline \hline
10000&	247.2&	3.8&	6.3& 85.7&	11.7 \\
50000	&      7661.0&	18.6& 37.2&		128.9&	62.6 \\
100000&	32933.2&	38.7& 80.3&		191.9&	133.7 \\ \hline
    \end{tabular}
    \label{table:running_n^5}
 }
\end{table}

We tested all programs on strings
of lengths $10000$, $50000$, and $100000$, all from an alphabet of size $|\Sigma|=10$.
Table~\ref{table:running_n^5} shows the results.
From Table~\ref{table:running_n^5},
we can confirm that \textsf{Traverse} is the fastest,
while \textsf{Na\"ive} is by far the slowest.
We further tested the algorithms 
on larger strings with $|\Sigma|=10$.
In this comparison, we excluded \textsf{Na\"ive} as it is too slow. 
The results are shown in Figure~\ref{figure:running_10^6}.
As one can see, \textsf{Traverse} was the fastest for all lengths.
We also conducted the same experiments varying alphabet sizes
as $2$, $4$, and $20$, and obtained similar results 
as the case of alphabet size $10$.

To verify why \textsf{Traverse} runs fastest,
we measured the average numbers of suffix array entries
which are traversed, per pivot and output (i.e., canonical longest SAGP).
Figure~\ref{figure:traverse} shows the result.
We can observe that although in theory $O(n)$ entries
can be traversed per pivot and output for a string of length $n$,
in both cases the actual number 
is far less than $n$ and grows very slowly as $n$ increases.
This seems to be the main reason why \textsf{Traverse}
is faster than \textsf{RB tree}, \textsf{vEB tree}, and \textsf{Y-fast trie}
which use sophisticated but also complicated predecessor/successor data structures.

\begin{figure}[t]
 \begin{minipage}{0.48\hsize}
  \begin{center} 
    \includegraphics[width=65mm]{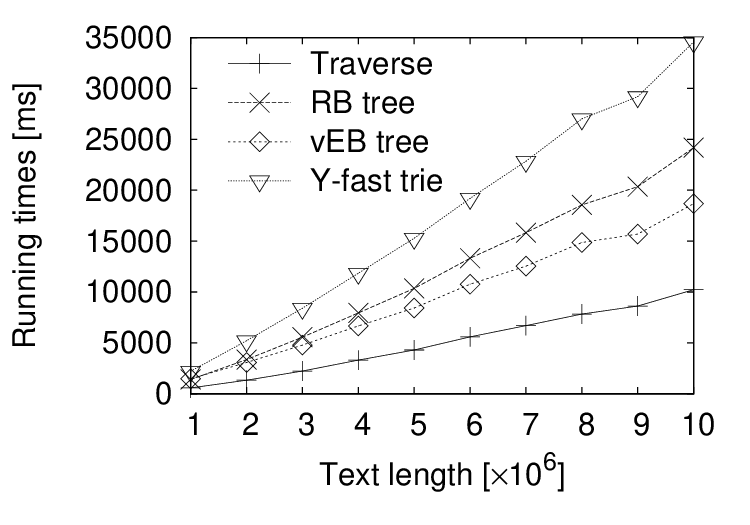}
    \end{center}
	\caption{Running times (in milli-sec.) on randomly generated strings of length from $1000000$ to $10000000$ with $|\Sigma|=10$.}
	\label{figure:running_10^6}
 \end{minipage}
 \hfill
 \begin{minipage}{0.48\hsize}
  \begin{center}
    \includegraphics[width=65mm]{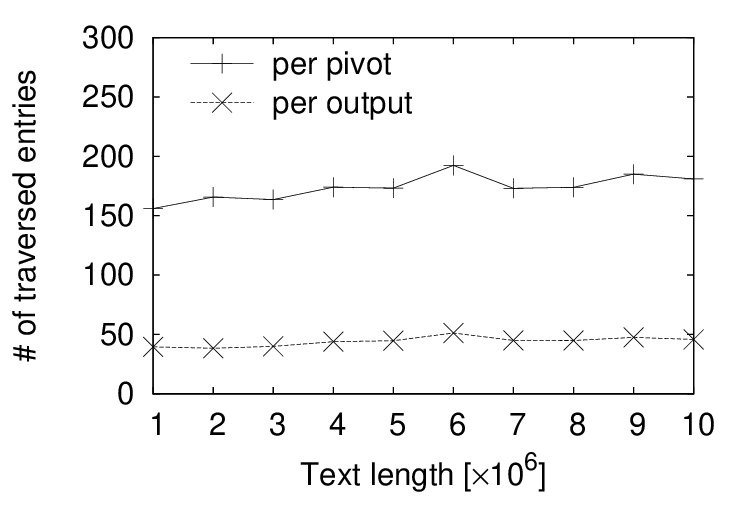}
    \end{center}
   \caption{Average numbers of traversed entries of suffix array per pivot and per output
   	on randomly generated strings.}
   \label{figure:traverse}
 \end{minipage}
\end{figure}


\section{Conclusions and future work}
We proposed several algorithms to compute longest single-arm-gapped palindromes (SAGPs)
for all pivots in a given string $T$ of length $n$.
For the type-1 longest SAGPs,
we presented an $O((n + \occ_1) \log\log n)$-time algorithm
which is based on a suffix array
and a dynamic predecessor/successor data structure,
where $\occ_1$ is the number of type-1 longest SAGPs to output.
We also presented an $O(n+\occ_1)$-time algorithm
based on the suffix tree.
For the type-2 longest SAGPs,
we proposed an $O(n + \occ_2)$-time algorithm,
where $\occ_2$ is the number of type-2 longest SAGPs to output.
Combining the last two results, we obtained an optimal $O(n + \occ)$-time
algorithm for computing all longest SAGPs,
where $\occ = \occ_1 + \occ_2$ is the total number of outputs.
We performed experiments to compare practical performances
of our algorithms for finding type-1 longest SAGPs;
the na\"ive algorithm, the $O(n^2)$-time suffix array based algorithm,
and the improved suffix array based algorithm with 
several kinds of predecessor/successor data structures.

Our future work includes the following:
Is there a linear $O(n + \occ_1)$-time algorithm for
finding type-1 longest SAGPs using a suffix array
(plus some auxiliary arrays), rather than a suffix tree?
Our current best suffix-array-based algorithm
uses $O((n + \occ_1) \log\log n)$ time.
Another question is how many longest SAGPs can be contained in a string.
Finding non-trivial upper bound and/or lower bound for $\occ$
remains open.

\section*{Acknowledgments}

The research of Shintaro Narisada, Kazuyuki Narisawa, and Ayumi Shinohara was supported by 
JSPS KAKENHI Grant Numbers JP15H05706, JP24106010, and ImPACT Program “Tough Robotics Challenge” of Japan Science and Technology Agency.
The research of Diptarama Hendrian is supported by Tohoku University Division for  Interdisciplinary Advance Research and Education.
The research of Shunsuke Inenaga was supported by JSPS KAKENHI Grant Numbers JP26280003 and JP17H01697.

The authors thank anonymous referees for their useful suggestions for improving the quality of the paper and for pointing out some errors in the previous version.

\bibliographystyle{plain} 
\bibliography{ref}

\end{document}